\documentclass[lettersize,journal]{IEEEtran}

\makeatletter
\def\endthebibliography{%
	\def\@noitemerr{\@latex@warning{Empty `thebibliography' environment}}%
	\endlist
}
\makeatother

\usepackage{amssymb,amsmath,amsthm}
\usepackage{cases}
\usepackage[linesnumbered,lined,boxed,commentsnumbered,ruled]{algorithm2e}	

\usepackage{algorithm2e}
\usepackage{algorithmic} 
\usepackage{booktabs}
\usepackage{graphicx}
\usepackage{float}
\usepackage{float}
\usepackage{subfigure}
\usepackage{stmaryrd}
\SetSymbolFont{stmry}{bold}{U}{stmry}{m}{n}
\usepackage{multirow,url}
\usepackage{color,cite}
\usepackage{cite}

\usepackage{tikz,pgf}
\usetikzlibrary{fadings,shapes,arrows,shadows}
\usepackage{makecell}
\usepackage{array}
\usepackage{overpic} 
\usepackage{subfigure}
\usepackage{url}
\usepackage{float}
\usepackage{booktabs} 
\usepackage{amsmath}
\usepackage{amsfonts}
\usepackage{mathrsfs}
\usepackage{amsfonts,amsthm,bm,bbm}
\usepackage{pifont}
\usepackage{xcolor}

\newlength{\thickarrayrulewidth}
\newtheorem{lemma}{Lemma}

\newtheorem{definition}{Definition}
\newtheorem{theorem}{Theorem}

\newtheorem{game}{Game}

\theoremstyle{plain}

\newcommand{\com}[1]{} 
\renewcommand{\qedsymbol}%
{\rule{1ex}{1.5ex}}

\def\E{\mathcal{E}}

\def\G{\mathcal{G}}
\def\M{\mathcal{M}}
\def\N{\mathcal{N}}
\def\K{\mathcal{K}}

\definecolor{mycolor1}{RGB}{5,40,150}

\IEEEoverridecommandlockouts

\begin{document}

\title{Joint Communication and Computation Scheduling for MEC-enabled AIGC Services: A Game-Theoretic Stochastic Learning Approach\\
}

\author{\IEEEauthorblockN{Huaizhe Liu, Xinyi Zhuang, Jiaqi Wu, Yuan Luo, Bin Cao, and Lin Gao \vspace{-7mm}}
\thanks{This work was supported in part by the Natural
Science Foundation of Guangdong Province
(Grant No. 2024A1515010178), in part
by the Shenzhen Science and Technology Program
(Grant No. GXWD20231129103946001, 
KJZD20240903095402004, and JCYJ20230807114300001), and in part by the National Natural Science Foundation of China (Grant No. 62472367).}
\thanks{H. Liu, X. Zhuang, J. Wu, B. Cao, and L. Gao are with the School of Information Science and Technology, and the Guangdong Provincial Key Laboratory of Aerospace Communication and Networking Technology, Harbin Institute of Technology, Shenzhen, China. 
Y. Luo is with the School of Science and Engineering, The Chinese University of Hong Kong, Shenzhen, China. 
L. Gao is also with the Shenzhen Loop Area Institute, Shenzhen, China. 
(e-mail: gaol@hit.edu.cn) 
\emph{(Corresponding Author: Bin Cao, Yuan Luo, and Lin Gao)}}
\thanks{Part of the results have been published in IEEE WiOpt 2024 \cite{liu_wiopt2024}.}
}

\maketitle

\begin{abstract}
Artificial Intelligence Generated Content (AIGC) powered by Generative Diffusion Models (GDMs) has emerged as a transformative paradigm for automated content creation. 
To satisfy the stringent latency requirements of AIGC services in many edge intelligence scenarios (e.g., smart cities), Mobile Edge Computing (MEC) provides critical computational support by deploying GDMs at edge servers (ES) close to end users. 
This paper investigates an MEC-enabled AIGC network comprising multiple ES, wireless access points (APs), and mobile users (UEs) with heterogeneous latency and accuracy demands. 
We formulate a \emph{Joint Communication Association and Computation Offloading (JCACO)} game, where each UE strategically selects its serving AP, ES, and inference steps to minimize the overall service completion time while meeting accuracy constraints. 
The problem is challenging due to the network dynamics and the incomplete information. 
We prove that the JCACO game is  a \emph{potential game} under both complete and stochastic information settings, ensuring the existence of Nash Equilibrium (NE) in both cases. 
To derive the NE efficiently, we develop a distributed \emph{Multi-Agent Stochastic Learning} (MASL) algorithm that provably converges to the NE with strict performance guarantees. 
Unlike conventional best-response schemes, MASL requires neither the knowledge of other players' strategies nor global network information, making it fully distributed and adaptive to dynamic environments. 
We further provide a strict theoretical convergence analysis for MASL by using \emph{Ordinary Differential Equations} (ODEs).
Simulation results demonstrate that MASL significantly reduces service completion time compared with benchmark methods while satisfying accuracy constraints, confirming its effectiveness and practicality for real-world MEC-enabled AIGC networks.~~~~~ 
\end{abstract}


\section{Introduction}\label{section:introduction}
\subsection{Background and Motivations}
Nowadays, AI-Generated Content (AIGC) based on Generative Diffusion Models (GDMs) has emerged as a promising paradigm of content generation, revolutionizing the creation of diverse contents and driving significant technological advancements.
To generate diverse contents, GDMs utilize a stochastic denoising process to transform simple noise distributions into complex data distributions \cite{xu_aigcsurvey2024}.
This process occurs gradually through a series of small, reversible denoising process (called \emph{inference steps}).
At each step, GDMs predict the noise component that needs to be subtracted to move towards the desired data distribution.
Due to the scalability with large datasets and complex data distributions, GDMs are well-suited for various tasks across different domains, such as image synthesis and video generation \cite{croitoru_tpami_2023}.

The iterative denoising process in GDMs is computationally intensive, posing challenges in its  application and deployment in many practical scenarios,  especially concerning latency and resource consumption. 
Typically, GDMs produce content through a reverse diffusion process that requires significant computational power and extended processing time. 
Some prior researches aimed to alleviate this complexity. 
For instance, in \cite{chuang_cvpr_2022} and \cite{rombach_cvpr_2022}, Chung \emph{et al.} and Rombach \emph{et al.} have modified inference mechanisms to enhance performance while reducing inference steps. 
In \cite{xu_iccv_2023}, Xu \emph{et al.} introduced the Versatile Diffusion (VD) network to lower computational demands. 
However, in the above works, GDMs are often deployed on remote cloud servers due to the considerable model size and computational requirements, which introduces substantial communication overhead and potential delays in delivering generated content to end-users.

\emph{Mobile Edge Computing (MEC)} \cite{MEC01, MEC02, MEC03} has emerged as a crucial enabler for GDM-based AIGC services, addressing the challenges of latency and resource consumption by moving computation closer to end-users.
Many existing works in this field have primarily focused on the joint scheduling of service management, model inference, computation task offloading, and resource allocation. 
For example, in \cite{du_tmc2024, zhang_twc2025, xu_twc2025, gao_tmc2025, zhang_iotj2024, yang_tmc2025, feng_tccn2025, zhuang_infocom2025, wujiaqi_tmc2025, liuy_tccn2025,jia_TWC2026,jia_comm2026,mao_jsac2025,mao_iotj2026,duj_tmc2024,tangc_tmc2025}, authors 
employed optimization methods based on complete information, aiming to enhance the Quality of Experience (QoE) for UEs.
However, complete information are difficult to obtain in practice due to the time-varying environment (i.e., UEs' inference tasks request).
To address this, authors in \cite{xux_tccn2025} and \cite{feng_infocom} applied stochastic optimization   or dynamic learning approaches under incomplete information, aiming to optimize overall system performance. 
However, the self-interested behavior of rational UEs leads to competition for limited resources, resulting in degraded performance, service outages, and increased waiting times. 
Therefore, purely optimization-based approaches are impractical in real-world scenarios.

\textcolor{black}{Game-theoretic methods are well-suited for modeling strategic interactions among multiple users in a decentralized MEC environment. 
	In game theoretic frameworks, UEs act as rational game players, strategically decide on resource allocation and task offloading to optimize their individual utility while accounting for others' actions \cite{liu_tmc2025,wen_iotj2024,ji_tmc2025,ye_tvt2025,lai_mwc2024,wang_ton2023,lai_tvt2025,lei_tmc2026,liu_tccn2025}.}
In our previous work \cite{liu_wiopt2024}, we studied a joint communication and computation scheduling problem in MEC-enabled AIGC network, and designed a distributed algorithm to minimize the task completion time for AIGC service users. 
However, the above works ignored the inherent stochasticity of  environment, such as the dynamic states of UEs and the time-varying conditions of wireless channels. 
Furthermore, they typically assume that UEs can fully observe and react to the strategies of others in a best-response manner, a condition that is rarely met in practical scenarios. 

To better capture the dynamic nature of the network environment, we propose a novel \emph{stochastic game-theoretic framework} to model and analyze the decision making process of UEs under incomplete information.
To derive the game equilibrium in such a stochastic environment, we propose a multi-agent \emph{stochastic learning} algorithm, which can achieve the Nash Equilibrium (NE) of the stochastic game, without requiring the knowledge of other players' strategies nor the complete network information. 

\vspace{-2mm}
\subsection{Solution and Contributions}
Specifically, we consider a general MEC-enabled mobile AIGC network scenario as shown in Fig.\ref{fig:system}, which consists of multiple edge computing servers (ES) offering GDM-based AIGC services, multiple access points (APs) providing wireless communication services, 
and multiple UEs requesting AIGC services from ES via APs.
Each ES is deployed with a specific GDM with a particular size and capability.
Different GDMs (on different ES) may have different sizes and capabilities, thus can provide AIGC services with different qualities.
APs are located at different locations and can provide communication services for UEs in a pre-defined coverage area using a limited communication bandwidth.
UEs are running different AIGC applications with diverse latency and accuracy requirements, and can request AIGC services from different ES through different APs, considering the heterogeneous capability of GDMs and the different locations of APs. 
We refer to a UE as an \emph{active} UE if it is
requesting AIGC services, and as an \emph{inactive} UE otherwise.

Fig.~\ref{fig:system} shows an example of such an MEC-enabled mobile AIGC network, 
where UEs 1–4 execute image generation tasks (e.g., using DALL\textperiodcentered E 3\footnote{https://dalle3.ai/} and Midjourney\footnote{https://www.midjourney.com/}), and UEs 5 and 6 perform video generation applications (e.g., ImagineArt\footnote{https://www.imagine.art/}).
ES 1, 2, and 3 are equipped with different GDMs, each capable of generating different types of content with varying quality. 
In this example, UE 1 requests AIGC service from ES 2 via the communication link provided by AP 1, and 
UE 4 requests AIGC service from ES 1 via the communication link provided by AP 3. 
This example also illustrates that the quality of generated contents improves as the number of inference steps increases, demonstrating a direct trade-off between content quality and computational effort.

Clearly, when a UE requests AIGC services, it needs to select an appropriate AP for content transmission and an suitable ES (or GDM) for content generation. 
Note that UEs choosing the same AP or ES will compete with each other for the limited communication bandwidth or computation resource.
Moreover, as shown in many existing studies (e.g.,\cite{yang_tmc2025,zhuang_infocom2025}, and \cite{xux_tccn2025}), the quality of generated content generated by GDMs is highly dependent on the number of inference steps.
Specifically, a larger number of inference steps yields better content quality but incurs higher computational costs, as illustrated in Fig. \ref{fig:system}.
Therefore, UEs must carefully determine the number of inference steps to balance content quality and computational cost. 
To address these concerns, we focus on studying the following two key problems in a decentralized MEC-enabled AIGC network: 
\begin{enumerate}[]
	\item \textit{Communication Association Problem}: How to determine the most appropriate AP for each UE, considering factors such as the locations of UEs and APs, the service latency requirements of UEs, and the communication workloads of APs?  
	\item \textit{Computation Offloading Problem}: How to determine the best GDM (ES) and the best number of inference steps for each UE, considering factors such as the model capabilities of GDMs, the service quality requirements of UEs, and the computation workload of ES?
\end{enumerate}

\begin{figure}[t]
	\centering
	\includegraphics[width=0.9\linewidth]{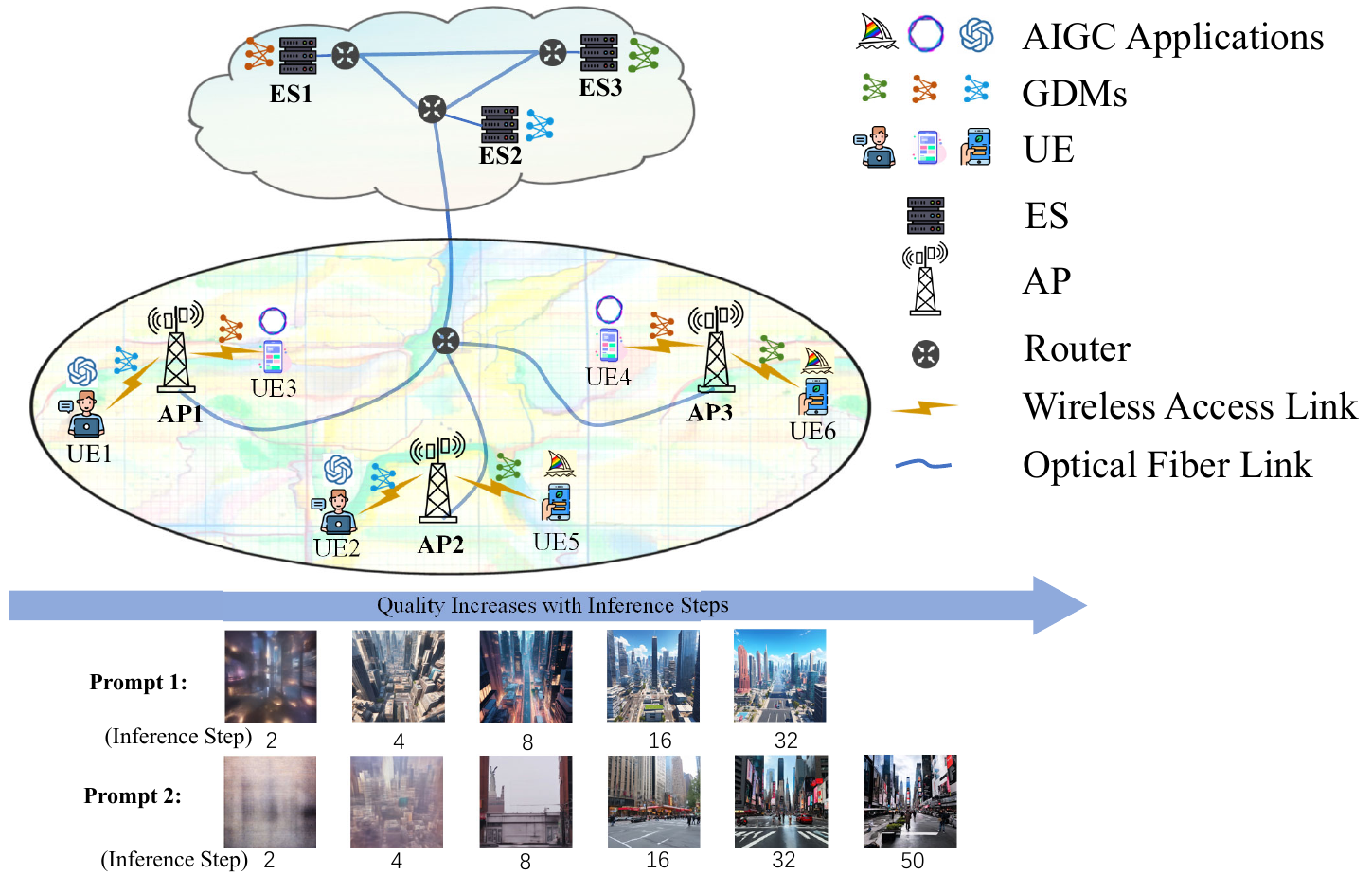}
	\vspace{-3mm}
	\caption{An example of MEC-enabled GDM-based AIGC Network.}
\label{fig:system}
\vspace{-3mm}
\end{figure}

The problems are challenging due to the decentralized nature of the system and the self-interested behavior of UEs,  thereby requiring robust solutions to ensure efficient, stable, and self-enforcing outcomes. 
Game theory \cite{GAME01, GAME02, GAME03} provides a powerful theoretic framework for analyzing the behaviors of self-interested players in decentralized and autonomous environments. 
Therefore, we utilize game theory to tackle the aforementioned problems in decentralized MEC-enabled AIGC networks.  
More specifically, we formulate a non-cooperative game among UEs, called \emph{Joint Communication Associating and Computation Offloading (JCACO)} game. 
In this game, each UE acts as a game player, 
determining  the best AP for communication, the best ES for computation, and the optimal number of inference steps for the GDM (on the chosen ES), aiming to minimize inference time while satisfying accuracy requirements. 

We analyze the JCACO game in two different information scenarios: \emph{complete information} and \emph{stochastic information}. 
In the complete information scenario, UEs can observe the complete activity states of other UEs (i.e., complete information)  during the game. 
In contrast, in the stochastic information scenario, UEs know only the probability distributions of other UEs' activity states (i.e., stochastic information), without knowing the complete activity information. 
We prove that the JCACO game is \emph{potential game} \cite{potential_game} in both information scenarios, thus ensuring the existence of Nash equilibrium (NE) in both scenarios. 
We further propose a distributed \emph{Multi-agent Stochastic Learning (MASL)} algorithm that converges to the NE with strict performance guarantee under stochastic information. 
We further provide a strict theoretical convergence analysis
for MASL by using Ordinary Differential Equations (ODEs). 
Unlike traditional best response algorithms in game theory, our proposed MASL algorithm requires neither the knowledge of other players' strategies, nor the complete system information, making it fully distributed and well-suited for dynamic environments. 
To sum up, the main contributions of this work are summarized as follows:
\begin{itemize}
\item[$\bullet$] \textit{Novel Model Scenario:}
We study a novel MEC-enabled GDM-based AIGC network, which incorporates heterogeneous GDMs with varying capabilities and heterogeneous AIGC users with diverse QoE  requirements.
We explore the joint communication and computation scheduling problem for UEs in a practical scenario with decentralized, autonomous, and self-interested UEs.

\item[$\bullet$] \textit{Novel Game-Theoretic Analysis:}
We formulate a non-cooperative game, \emph{JCACO}, to model the communication and computation competition among UEs. 
We prove that the JCACO game is a \emph{potential game} in both complete and stochastic information scenarios, guaranteeing the existence an NE. 
We further propose a stochastic learning-based algorithm, \emph{MASL}, to achieve NE with rigorous performance guarantees under stochastic conditions.

\item[$\bullet$] \textit{Performance Evaluations and Insights:}
Extensive simulation results show that the proposed MASL algorithm significantly reduces service completion time while meeting accuracy requirements, in comparison to baseline solutions.
These results highlight the effectiveness of MASL in real-world AIGC network scenarios.
\end{itemize}

The rest of the paper is organized as follows. 
In Section \ref{section:System_model}, we present the system model. 
In Sections \ref{section:complete_info_game} and \ref{section:sto_info_game}, we formulate and analyze the JCACO game under complete information and stochastic information, respectively. 
In Section \ref{section:algorithm}, we design the MASL algorithm for the stochastic information scenario. 
We present simulation results in 
Section \ref{section:simulations}, and finally conclude in Section \ref{section:Conclusion}. 
\section{Related Works}\label{section:Related_works}

\begin{table}[t]
	\begin{center}\caption{A comparison with existing works}\label{table:reference_compare}
		{\fontsize{7.2}{7.2}\selectfont
			\begin{tabular}{|c|c|c|c|c|}
\hline
				\multirow{3}{*}{\textbf{References}} & \multicolumn{2}{c|}{\textbf{Decision Variables}} & \multirow{3}{*}{\textbf{Information}} & \multirow{3}{*}{\textbf{Method}}\\
				\cline{2-3}
				 & \makecell[c]{ \emph{Communication} \\ \emph{Scheduling}} & \makecell[c]{\emph{Computation} \\ \emph{Scheduling}} &  &     \\ 
				\hline
                \hline
\cite{du_tmc2024,zhang_twc2025,gao_tmc2025,yang_tmc2025,feng_tccn2025} &  & \checkmark  & \multirow{3}{*}{Complete} &  \multirow{4}{*}{Optimization}    \\ 
				\cline{1-3}
				\cite{liuy_tccn2025,jia_TWC2026,jia_comm2026} & \checkmark &  &  &       \\
				\cline{1-3}
				\cite{mao_jsac2025, mao_iotj2026,duj_tmc2024,tangc_tmc2025,
					 xu_twc2025,zhang_iotj2024,zhuang_infocom2025,wujiaqi_tmc2025} & \checkmark & \checkmark  &  &     \\
				\cline{1-4}
				\cite{xux_tccn2025,feng_infocom} &  & \checkmark & Stochastic  &     \\
						\hline
                \hline
		\cite{liu_tmc2025,wen_iotj2024,ji_tmc2025,lai_mwc2024,ye_tvt2025} &  & \checkmark & \multirow{2}{*}{Complete} &       \\
				\cline{1-3}
				\cite{wang_ton2023, lai_tvt2025,lei_tmc2026,
					liu_wiopt2024,liu_tccn2025} & \checkmark & \checkmark &  &\multirow{3}{*}{Game} \\
				\cline{1-4}
				\cite{SLA1,SLA2} & \checkmark &  & \multirow{2}{*}{Stochastic} & \\
				\cline{1-3}
				\cite{SLA3} &  & \checkmark &  &      \\
				\hline
                \hline
				\textbf{Our work} & \checkmark & \checkmark & $\textbf{Stochastic}$ & \textbf{Game}  \\
		\hline
			\end{tabular}
		}
	\end{center}
\end{table}

Recently, MEC-enabled AIGC services have attracted considerable attention from both industry and academia. 
In this section, we will review existing works in this field from two different modeling perspectives: optimization modeling and game modeling, under two different information scenarios: complete information and stochastic information.
\subsection{Optimization Modeling under Complete Information}
Many studies in MEC-enabled AIGC networks assume the availability of complete information, enabling centralized decision-making for global optimization \cite{zhang_twc2025,xu_twc2025,gao_tmc2025,du_tmc2024,zhang_iotj2024,yang_tmc2025,feng_tccn2025,zhuang_infocom2025,wujiaqi_tmc2025,liuy_tccn2025,jia_TWC2026,jia_comm2026,mao_jsac2025,mao_iotj2026,duj_tmc2024,tangc_tmc2025}.
Within this framework, it is assumed that all relevant system information, including network status and UE requirements, is fully accessible.
\textcolor{black}{In \cite{du_tmc2024,zhang_twc2025,gao_tmc2025,yang_tmc2025,feng_tccn2025}, the authors focus on computation resource scheduling for AIGC services, aiming to optimize model selection, denoising steps, and collaborative offloading to balance service quality with latency.
	In \cite{liuy_tccn2025,jia_TWC2026,jia_comm2026}, the authors focus on communication scheduling, aiming to optimize resource allocation and service quality in emerging wireless networks.
	Meanwhile, in \cite{xu_twc2025,zhang_iotj2024,zhuang_infocom2025,wujiaqi_tmc2025,mao_jsac2025,mao_iotj2026,duj_tmc2024,tangc_tmc2025}, the authors jointly optimize both communication and computation resources, addressing multi-objective optimization problems for AIGC task offloading with goals such as minimizing latency, improving energy efficiency, and achieving load balancing.}

However, the assumption of complete information often fails to capture the complexities of real-world deployments. In practice, system states and UE behaviors are inherently dynamic, making it difficult to guarantee complete visibility. 
As a result, models based on such assumptions may lack practical applicability.

\subsection{Optimization Modeling under Stochastic Information}
To address the limitations of complete information models, \cite{xux_tccn2025} and \cite{feng_infocom} have explored optimization under stochastic information. 
In \cite{xux_tccn2025}, Xu \emph{et al.} studied a content quality maximization problem in MEC-enabled AIGC networks, aiming to maximize content quality while adhering to latency and energy consumption constraints.
In \cite{feng_infocom}, Feng \emph{et al.} studied a novel online text-to-image AIGC request scheduling problem under a heterogeneous edge computing scenario with the objective of striking a balance between tardiness delay of AIGC requests and quality of generated results.

While these studies represent important progress, they largely adopt centralized optimization approaches that may not effectively capture the decentralized and autonomous nature of real-world edge networks. In practical scenarios,  self-interested UEs compete for limited resources. 
Unilateral optimization methods that neglect user interaction dynamics may lead to inefficient or unstable outcomes. 

\subsection{Game Modeling under Complete Information}
To model strategic interactions among UEs in MEC-enabled AIGC networks, game-theoretic methods have been widely adopted \cite{liu_tmc2025,wen_iotj2024,ji_tmc2025,ye_tvt2025,wang_ton2023,lai_tvt2025,lei_tmc2026,liu_tccn2025,lai_mwc2024}. 
To be specific, in \cite{liu_tmc2025}, Liu \emph{et al.} proposed a contest-based incentive mechanism to motivate edge devices to allocate resources effectively in image generation tasks.
In \cite{wen_iotj2024}, Wen \emph{et al.} designed an incentive mechanism for edge AIGC services in 6G IoT networks, incorporating a contract theory model to incentivize ASPs to provide AIGC services to clients.
In \cite{ji_tmc2025}, Ji \emph{et al.} formulated an auction framework to tackle the time-cumulative social cost optimization problem, aiming to achieve the trade-off between inference cost and accuracy.
In \cite{ye_tvt2025}, Ye \emph{et al.} proposed a two-stage resource allocation framework using generative diffusion models and contract theory to jointly optimize prompt, inference steps, computation resources, and transmission rate, in order to improve the quality of AI-generated content and reduce the latency.
In \cite{lai_mwc2024}, Lai \emph{et al.} designed a resource-efficient incentive mechanism for mobile AIGC networks in resource-constrained environments.
\textcolor{black}{In \cite{wang_ton2023}, Wang \emph{et al.} designed a Stackelberg game to model the leader-follower interactions between resource-constrained UAVs and cooperative ground vehicles.
	In \cite{lai_tvt2025}, Lai \emph{et al.} proposed a learning-based Stackelberg game framework between AIGC service providers and vehicular metaverse users, employing a Transformer-enhanced deep reinforcement learning algorithm to derive optimal pricing and computational resource allocation policies that balance service quality and provider profitability. }
In \cite{lei_tmc2026}, Lei \emph{et al.} formulated a joint scheduling problem as a multi-node non-cooperative game, and designed a Federated MADDPG algorithm to incentivize privacy-preserving collaborative offloading among distributed edge servers.
In \cite{liu_tccn2025}, Liu \emph{et al.} ‌proposed a GDM-based multi-modal semantic communication framework (GM-SemCom) to optimize resource allocation in mobile AIGC networks, specifically developing a Stackelberg game model with a GDM-based algorithm to address the resource trading problem.
In our early work \cite{liu_wiopt2024}, we studied a joint communication and computation scheduling problem in MEC-enabled AIGC network, and designed a distributed algorithm to minimize the task completion time for AIGC service users.
These game-theoretic frameworks allow for the analysis of equilibrium behaviors in distributed decision-making environments.

While the above game-theoretic methods offer a more accurate representation of real-world scenarios, they assume complete information about the system, which is not always feasible. 
In practical MEC-enabled AIGC systems, the utility of each user fluctuates due to dynamic network conditions, making static models based on complete information less effective. This highlights the need for models that can handle such uncertainties and adapt to changing environments.

\subsection{Game Modeling under Stochastic Information}
To overcome the limitations of static game-theoretic models under complete information, the authors of \cite{SLA1, SLA2, SLA3} have applied stochastic game theory to analyze strategic behavior in dynamic, uncertain environments. 

In \cite{SLA1} and \cite{SLA2},  authors employed stochastic game frameworks to study channel allocation and interference mitigation in dynamic wireless networks.
In \cite{SLA3}, Zheng \emph{et al.} modeled multi-user computation offloading in mobile cloud computing, where UE activity and wireless conditions exhibit random fluctuations. These studies offer a more realistic representation of dynamic interactions in MEC-enabled AIGC services, highlighting the need for coordinated management of communication and computation resources.

Although these approaches provide theoretical insight, their applicability to AIGC scenarios remains limited. Most of these studies focus either on communication or computation resource management, without considering the interdependencies between the two. Moreover, they typically overlook the trade-offs between inference accuracy and computational latency, which are crucial to the performance of AIGC services. 

\subsection{New Features of Our Work}
In this work, we address the joint communication and computation scheduling problem in MEC-enabled AIGC networks by incorporating stochastic information and self-interested user behaviors. 
We propose a stochastic game-theoretic approach that effectively captures the interdependencies between user decisions and the evolving system dynamics. 
Unlike previous studies, our model jointly optimizes both communication and computation resources, ensuring that all factors influencing AIGC service performance are considered and improved. 
To summarize, our work provides a significant advancement over previous studies by integrating stochastic game theory into the joint scheduling of communication and computation resources for MEC-enabled AIGC services. 
This approach better reflects the complexities of real-world, decentralized environments and offers a novel solution for optimizing both task completion time and service quality in dynamic environments. 
For clarity, we summarize the key features of our approach and existing methods in Table \ref{table:reference_compare}. 
\section{System Model}\label{section:System_model}
\begin{table}[t]
	\caption{Main notations in this work}
	\label{table:main_notations} 
	\hspace{-3mm}
	\begin{tabular}{l|p{20em}}
		\hline
		\textbf{Symbols} & \textbf{Description} \\
		\hline
		\hline
		$\mathcal{M}$ & The set of APs, with $m \in \mathcal{M}$ denoting an particular AP\\
		$\mathcal{K}$ & The set of ES (GDMs), with $k \in \mathcal{K}$ denoting an particular ES (GDM)\\
		$\mathcal{N}$ & The set of UEs,   $n \in \mathcal{N}$ denoting a particular UE \\
		$R_{nm}$ & The channel capacity of the wireless link between UE $n$ and AP $m$ \\
		$T^{\mathrm{Acc}}_{nm}$ & The transmission time required for transmitting the
		generated content to UE $n$ from AP $m$ \\
		$T^{\mathrm{Comp}}_{nm}$ & The computation time required by the GDM on ES $k$ for generating content of UE $n$ \\
		$\boldsymbol{\omega}  $ & The complete information realization 
		$\boldsymbol{\omega} \triangleq (\omega_n \in \{0,1\}, \forall n\in \N) $ \\
		$\mathbf{X} \triangleq(x_{nm},\forall n,m)$ & The UE-AP association relationships between all APs and all UEs\\
		$\mathbf{Y}\triangleq(y_{nk},\forall n,k)$ & The offloading matrix consisting
		of the offloading indicators between all ES and all
		UEs  \\
		$\mathbf{D}\triangleq(d_{nk},\forall n,k)$ & The inference step matrix consisting of the   inference steps used by all GDMs for all UEs. \\
		$T^{\mathrm{Acc}}_{n}(\mathbf{X};\boldsymbol{\omega})$ & The actual transmitting completion time of UE $n$ under complete information realization \\ 
		$T^{\mathrm{Comp}}_{n}(\mathbf{Y,D};\boldsymbol{\omega})$ & The actual computation completion time of UE $n$ under complete information realization\\
		$T^{\mathrm{Total}}_{n}(\mathbf{X,Y,D};\boldsymbol{\omega})$ & The actual computation completion time of UE $n$ under complete information realization\\
		$\overline{T}^{\mathrm{Acc}}_{n}(\mathbf{X})$ & The expected transmitting completion time of UE $n$ under stochastic information \\ 
		$\overline{T}^{\mathrm{Comp}}_{n}(\mathbf{Y,D})$ & The expected computation completion time of UE $n$ under stochastic information \\
		$\overline{T}^{\mathrm{Total}}_{n}(\mathbf{X,Y,D})$ & The expected computation completion time of UE $n$ under stochastic information \\
		\hline
	\end{tabular}
\end{table}
\subsection{Network Model}
As shown in Fig.~\ref{fig:system}, we consider an MEC network consisting of a set $\M=\{1,2,...,M\}$ of wireless APs, a set $\K=\{1,2,...,K\}$ of ES, and a set $\N= \{1,2,...,N\}$ of UEs.
Each AP is located at a specific location, providing network access service for UEs within its signal coverage area. 
Each ES is deployed with a trained GDM and can offer GDM-based AIGC services. 
For notational convenience, we refer to the GDM deployed on ES $k$ as GDM $k$.
APs and ES are connected via routers and optical fiber links, enabling efficient communication.
UEs can request AIGC services from ES through nearby APs. 
Different from cloud computing servers(e.g., Google Data Centers), ES are deployed at the network edge that approximates to UEs, and thus can provide AIGC services for UEs with low round-trip time (RTT). \par
As in many existing studies (e.g., \cite{SLA1,SLA2,SLA3}), we adopt a time-slotted structure, where the entire time period is divided into multiple \emph{time slots} of equal time length (e.g., 1 second). 
We assume a \emph{quasi-static} network, where the network state, such as the locations of UEs and wireless channel conditions, keeps unchanged within each time slot, but varies randomly across different time slots. 
For convenience, the key notations used in this work are summarized in Table~\ref{table:main_notations}.

\subsection{UE Activity State Model}
For clarity, we refer to a UE as an \emph{active} UE if it is requesting AIGC services, and as an \emph{inactive} UE otherwise. The activity state of each UE, which indicates whether the UE is actively requesting AIGC services, plays a crucial role in determining its demand for services from ESs. 

It is important to note that, in practice, UEs may not always be active in every time slot. 
Let $\omega_{n} \in \{0,1\}$ represent the activity state of UE $n$ during a specific time slot.
Specifically,  $\omega_{n} = 1$ indicates that UE $n$ is active, meaning it is requesting AIGC services during that time slot, while $\omega_{n} = 0$ indicates that UE $n$ is inactive and not making any service request. 
This binary representation of the activity state simplifies the modeling of service requests and facilitates future analysis of network behavior and resource allocation. 
For notational convenience, we define $\N_A$ as the set of active UEs, i.e., $\N_A \triangleq \{n \in \N | \omega_{n} = 1 \}$. 

\textcolor{black}{In practice, UE requests for AIGC services may exhibit temporal and spatial correlations, especially in hotspot scenarios or during specific peak periods. To maintain analytical tractability, and under the quasi-static setting considered in this work, we assume that within the considered decision horizon the distribution of UE activity states remains stationary. Following \cite{SLA2}, we then model the activity state $\omega_n$ of each UE $n$ as an on-off \emph{Bernoulli random variable} with parameter $p_n$.}
That is~\footnote{
	From a mathematical perspective, $p_n$ is continuous. However, in practical network systems, there exists a minimum non-zero measurement, which we denote as $\Delta^{\mathrm{prob}}$.}:
\begin{equation} \label{active_UE} \small
	\begin{cases}
		\Pr(\omega_n = 1) = p_n, \\
		\Pr(\omega_n = 0) = 1 - p_n,
	\end{cases}
\end{equation}
where $p_n \in [0,1]$ is the active probability of UE $n$.
For notational convenience, we introduce a random vector $\boldsymbol{\omega}$ to denote the activity states of all UEs, that is,
\begin{equation} \label{active_state_UE}
	\boldsymbol{\omega} \triangleq (\omega_n)_{n \in \mathcal{N}}.
\end{equation}

A specific realization of $\boldsymbol{\omega}$ is referred to as a \emph{sample}, and the set of all possible realizations constitutes the sample space, denoted by $\Omega$. 
We denote by $p(\boldsymbol{\omega}) = \Pr(\boldsymbol{\omega})$ the joint probability mass function of the activity state $\boldsymbol{\omega}$. 
Under the independence assumption across UEs, this joint probability factorizes as:
\begin{equation}
	\begin{small}
		p(\boldsymbol{\omega}) = \prod_{n \in N_A} p_n  \cdot \prod_{n \in N/N_A} (1-p_n).
	\end{small}
\end{equation}

\subsection{AIGC Service Model}
We consider a general AIGC scenario where UEs request different AIGC services, and ES deploy different GDMs to fulfill these demands.
Each GDM is capable of providing a range of AIGC services, but the quality and cost of these services vary depending on the model’s capabilities and characteristics. 
For instance, a GDM specifically optimized for certain data modalities (e.g., images) may deliver superior output quality and lower computational overhead within its domain of specialization. 
However, when applied to content generation beyond its intended modality, its performance may degrade due to reduced output fidelity and diminished processing stability.

The performance of GDM-based AIGC services is influenced by several factors, including model size, training data, and the computational resources required for operation \cite{croitoru_tpami_2023,chuang_cvpr_2022,rombach_cvpr_2022}. 
Generally, larger models with more extensive training on diverse datasets tend to produce higher-quality outputs but also demand greater computational power and storage, leading to higher operational costs. 
In contrast, smaller, more specialized models may offer more cost-effective solutions for specific applications, although they often lack the versatility and performance of their larger counterparts.

As discussed in \cite{yang_tmc2025, feng_tccn2025}, the quality of the AIGC services provided by GDMs is typically evaluated by the \emph{average error of generated contents} (AEC). AEC is heavily influenced by the number of inference steps used by the GDMs in their reverse diffusion process for iterative denoising.
Specifically, let $E_{nk}$ denote the AEC of UE $n$'s content generated by GDM $k$, and let $d_{nk}$ denote the number of inference steps used by GDM $k$ for UE $n$'s content.
According to \cite{zhuang_infocom2025} and \cite{wujiaqi_tmc2025}, the AEC of UE $n$'s content generated by GDM $k$ can be evaluated through the following function: 
\begin{equation} \label{AEC}
E_{nk} = \epsilon_k^{\mathrm{fwd}} \cdot e^{-\gamma_{nk} \cdot d_{nk}}, \forall k\in \K, n\in \N,
\end{equation}
where $\epsilon_k^{\mathrm{fwd}}$ is a model-specific scaling factor evaluating the inherit inference capability of GDM $k$, and $\gamma_{nk}$ is an attenuation factor that characterizes the fitness of GDM $k$ for UE $n$'s content requirement. 
A larger $\gamma_{nk}$ leads to a smaller AEC $E_{nk}$, implying that GDM $k$ is more effective in completing the AIGC task for UE $n$.
Clearly, we consider a general scenario where different GDMs have different strengths in generating different types of contents.

The computational cost of generating AIGC content with a GDM is quantified by the total number of Floating Point Operations (FLOPs), which is directly proportional to the number of inference steps used.
Let $\xi_k$ represent the number of FLOPs required for each inference step of GDM $k$ (i.e., FLOPs per step).
The value of $\xi_k$ is related to the model size and architecture, with larger models and more complex architectures typically leading to a higher $\xi_k$.
The total number of FLOPs required by GDM $k$ to generate UE $n$'s AIGC content, denoted as $F_{nk}$, is then given by:
\begin{equation}
F_{nk} = \xi_{k} \cdot d_{nk}, \forall k\in \K, n\in \N.
\end{equation}

\subsection{Communication Model}\label{section:communication_model}
As shown in Fig.\ref{fig:system}, when a UE requests an AIGC service,
it needs to select an ES or GDM for content generation and an AP for data transmission.
Next, we will discuss the communication process, which mainly consists of the wireless communications between UEs and APs and the wired communications between APs and ES.

\subsubsection{Communication links between UEs and APs}
To capture the time-varying wireless channels, we assume that the channels between UEs and APs experience Rayleigh fading. 
Specifically, the instantaneous channel gain from UE $n$ to AP $m$ is given by $h_{nm} = (d_{nm})^{-\beta}\alpha_{nm}$, where $d_{nm}$ represents the distance between UE $n$ and AP $m$, $\beta$ denotes the path loss exponent, and $\alpha_{nm}$ represents the Rayleigh fading factor.
It is important to note that while $h_{nm}$ varies over time, it remains constant within a given time slot.
Let $R_{nm}$ denote the channel capacity of the wireless link between UE $n$ and AP $m$.
According to the Shannon-Hartley theorem, the channel capacity can be expressed as:
\begin{equation} \small \label{channel_capacity}
\begin{aligned}
	R_{nm} = W \cdot \log_{2}(1+\frac{\rho_{nm} \cdot h_{nm}}{N_0}),\forall m\in \M, n\in \N,
\end{aligned}
\end{equation}
where $\rho_{nm}$ denotes the transmission power for transmiting UE $n$'s data from AP $m$, and  $h_{nm}\triangleq \mathtt{d}_{nm} ^{-\theta}$ denotes the wireless channel gain. 
Here $\mathtt{d}_{nm}$ denotes the distance between UE $n$ and AP $m$, while $\theta$ denotes the path loss factor.
Next, the total transmission time required to transmit the generated content to UE $n$ from AP $m$  is derived as follows: 

\begin{equation} \small \label{transmission_delay}
\begin{aligned}
	T_{nm}^{\mathrm{Acc}} \triangleq  \frac{D_{n}}{R_{nm}},\forall m\in \M, n\in \N.
\end{aligned}
\end{equation}
where $D_n$ denotes the data size of the generated content for UE $n$. 
It is worth noting that the transmission time derived in (\ref{transmission_delay}) assumes that UE $n$ occupies the entire wireless transmission capacity of AP $m$.
In reality, when multiple UEs are connected to the same AP, they share the limited communication resources (e.g., bandwidth). 
In this work, we adopt a proportional fairness (PF) resource-sharing mechanism for communication resource, where all UEs connected to the same AP share the total communication resources in proportion to their transmission demands. 
Under this scheme, all UEs connected to the same AP will complete data transmission at approximately the same time.

Let  $x_{nm} \in \{0,1\}$ denote the UE-AP association indicator, where $x_{nm}=1$ indicates that UE $n$ is connected to AP $m$.
The actual transmission completion time for active UEs connected to AP $m$ is given by:
\begin{equation} \small \label{load_AP}
\begin{aligned}
	L_{m}(\mathbf{X};\boldsymbol{\omega}) = \sum_{n\in\N_{A}} x_{nm}\cdot  T_{nm}^{\mathrm{Acc}}, \forall m\in \M,
\end{aligned}
\end{equation}
where $\mathbf{X} \triangleq (x_{nm}, \forall n, m) $.
This completion time is equivalent to the total time required for AP $m$ to sequentially transmit data to all of its connected UEs.
Therefore, we refer to $L_m(\mathbf{X};\boldsymbol{\omega})$ as the \emph{load} of AP $m$, which accounts for the activeness of each connected UE.
The actual transmitting completion time of UE $n\in\N_{A}$ can be expressed as:
\begin{equation} \small \label{acc_delay}
\begin{aligned}
	T_{n}^{\mathrm{Acc}}(\mathbf{X};\boldsymbol{\omega}) = \sum_{m\in \M} x_{nm} \cdot L_m (\mathbf{X};\boldsymbol{\omega}), \forall  n\in \N_{A}.
\end{aligned}
\end{equation}

\subsubsection{Communication Links between APs and ES}
Following the approach used in previous studies(e.g., \cite{M.Tang_TMC2022} and \cite{greedy1}), we assume that the wired optical fiber links between APs and ES have infinite capacity. 
As a result, the transmission time between APs and ES is neglected.

\subsection{Computation Model}\label{section:computation_model}
We now discuss the computation process for each GDM on each ES.
Let $f_k$ denote the number of FLOPs (per unit of time) of ES $k$.
The total computation time required by the GDM on ES $k$ to generate content for UE $n$ can then be expressed as:
\begin{equation} \small \label{compute_delay}
\begin{aligned}
	T_{nk}^{\mathrm{Comp}} =  \frac{F_{nk}}{f_k},  \forall k \in \K, n \in \N,
\end{aligned}
\end{equation}
where $F_{nk} \triangleq \xi_k \cdot d_{nk}$ denotes the total number of FLOPs used by GDM $k$ to generate content for UE $n$.
Here, $d_{nk}$ is the number of inference steps that UE $n$ requires when using GDM $k$ to generate its content, and $\xi_{k}$ is the number of FLOPs consumed per inference step for GDM $k$.

Similar as in Section \ref{section:System_model}.D, we adopt a proportional fairness (PF)-based mechanism for computation resource. In this mechanism, all UEs offloaded to the same ES share the available computation resources in proportion to their individual computation demands, and thus will complete their content generation tasks at approximately the same time.

Let  $y_{nk} \in \{0,1\}$ denote the computation offloading indicator, where $y_{nk}=1$ indicates that UE $n$ offloads its task to ES $k$.
The actual computation completion time for active UEs offloaded to ES $k$ can then be written as:
\begin{align} \small \label{load_ES}
I_{k}(\mathbf{Y,D};\boldsymbol{\omega})	= \sum_{n\in \N_{A}} y_{nk} \cdot T_{nk}^{\mathrm{Comp}}, \forall k\in \K,
\end{align}
where $\mathbf{Y} \triangleq (y_{nk}, \forall n, k) $, and $\mathbf{D} \triangleq (d_{nk}, \forall n, k)$.

It follows that the completion time described above represents the total time required for ES $k$ to sequentially generate content for all offloaded UEs. 
Consequently, we refer to $I_k(\mathbf{Y,D;\boldsymbol{\omega}})$ as the \emph{load} of ES $k$, which depends on the activity of each UE.
Given the offloading vector $\boldsymbol{y}_n = (y_{n1},...,y_{nK})$
and the inference step vector $\boldsymbol{d}_n = (d_{n1},...,d_{nK})$, the actual computation completion time for UE $n \in \N_{A}$ can be expressed as:
\begin{equation} \small \label{comp_delay}
\begin{aligned}
	T_{n}^{\mathrm{Comp}}(\mathbf{Y,D;\boldsymbol{\omega}})  
	= \sum_{k\in \K} y_{nk} \cdot I_{k}(\mathbf{Y,D};\boldsymbol{\omega}), \forall n \in \N_{A}.
\end{aligned}
\end{equation}

\subsection{Problem Formulation}
To minimize the total service time for UEs while meeting the specified inference accuracy requirements, we formulate the Joint Computation Association and Computation Offloading (JCACO) problem, taking the activity states of all UEs into consideration.
The total service time for each UE consists of both the transmission completion time and the computation completion time.
Based on   (\ref{acc_delay}) and (\ref{comp_delay}), the total service latency for an active UE $n \in \N_{A}$, considering complete information about the activity of UEs, can be expressed as:
\begin{equation} \small \label{objective}
	\begin{aligned}
		T_{n}^{\mathrm{Total}}(\mathbf{X,Y,D;\boldsymbol{\omega}}) 
		= T_{n}^{\mathrm{Acc}}(\mathbf{X};\boldsymbol{\omega}) + T_{n}^{\mathrm{Comp}}(\mathbf{Y,D;\boldsymbol{\omega}}).
	\end{aligned}
\end{equation}

\subsubsection{Problem Formulation under Complete Information}
In the complete information scenario, we assume that the whole system information is available. 
Based on these information, we can formulate the JCACO problem, denoted as $\mathbf{P0}$:
\begin{align}
\mathbf{P0}\quad&\underset{\{ \mathbf{X,Y,D} \}}{\text{min}}~\sum_{n\in \N_{A}}T_{n}^{\mathrm{Total}}(\mathbf{X,Y,D};\boldsymbol{\omega}) \label{P0} \\ 	
\text {s.t.} \notag
&\textstyle ~\sum_{k=1}^{K}y_{nk} \cdot E_{nk}\leq \E_{n}^{0} \tag{\ref{P0}{a}}\label{P0a}, \forall n\in \N_{A}, \\
&\textstyle ~\sum_{m=1}^{M} x_{nm} = 1, \forall n\in \N_{A}, \tag{\ref{P0}{b}} \label{P0b}\\
&\textstyle ~\sum_{k=1}^{K} y_{nk}=1, \forall n\in \N_{A}. \tag{\ref{P0}{c}} \label{P0c} 
\end{align}

In $\mathbf{P0}$, the constraint in (\ref{P0a}) ensures that the achieved AEC for UE $n \in \N_{A}$ does not exceed the threshold $\E_{n}^{0}$, which reflects the QoE requirement.
The constraints in (\ref{P0b}) and (\ref{P0c}) ensure that each active UE is associated with only one AP and offloaded to only one ES.
Thus, $\mathbf{P0}$ represents a centralized optimization problem, aimed at optimizing the total service time for active UEs while satisfying their QoE requirements. \par
\subsubsection{Problem Formulation under Stochastic Information}
In the stochastic information scenario, we assume that the whole system information is not available, but the stochastic distribution information of system state is available. 
To account for the stochastic activeness dynamics of the UEs, we further reformulate $\mathbf{P0}$ as a stochastic optimization problem $\mathbf{P1}$:
\begin{align}
\mathbf{P1} 
&\underset{\{ \mathbf{X,Y,D} \}}{\text{min}}~ \mathbb{E}_{\boldsymbol{\omega}\in \Omega}\bigg[ \sum_{n\in \N}T_{n}^{\mathrm{Total}}(\mathbf{X,Y,D;\boldsymbol{\omega}})\bigg] \label{P1} \\ 	
\text{s.t.} \notag
&\textstyle ~\sum_{k=1}^{K}y_{nk} \cdot E_{nk}\leq \E_{n}^{0} \tag{\ref{P1}{a}}\label{P1a}, \forall n\in \N, \\
&\textstyle ~\sum_{m=1}^{M} x_{nm} = 1, \forall n\in \N, \tag{\ref{P1}{b}} \label{P1b}\\
&\textstyle ~\sum_{k=1}^{K} y_{nk}=1, \forall n\in \N. \tag{\ref{P1}{c}} \label{P1c} 
\end{align} 

In $\mathbf{P1}$, the objective is to minimize the expected total service time for all UEs under stochastic conditions, taking into account the variability of UE activities.

However, both $\mathbf{P0}$ and $\mathbf{P1}$ are centralized optimization problems, which may not be practical for implementation in decentralized, autonomous, and self-interested systems, especially considering the dynamic nature of UE activity. 
In a practical decentralized system, each UE seeks to minimize its own service time at the expense of global system efficiency. This self-interested behavior complicates the application of centralized optimization methods, as they require complete system-wide information.

To address these challenges, we will study both $\mathbf{P0}$ and $\mathbf{P1}$ from the game-theoretic perspective, where UEs are modeled as players in a game, adopting strategies that balance their individual objectives with the overall system efficiency. 
It's noteworthy that, $\mathbf{P0}$ and $\mathbf{P1}$ characterize two distinct information structures.

\textcolor{black}{In the complete information case ($\mathbf{P0}$), UE activity is known, leading to a static game under complete information. 
Conversely, the stochastic information case ($\mathbf{P1}$) relies on UE activity distribution, resulting in a stochastic game over system dynamics. 
Thus, the analysis of $\mathbf{P0}$ serves as a theoretical stepping stone to facilitate the analysis of the more general $\mathbf{P1}$.}
In the following, we will first analyze the game under complete information in Section \ref{section:complete_info_game}, and then analyze the stochastic game under stochastic information in Section \ref{section:sto_info_game}. 
\section{Game Formulation and Analysis under Complete Information} \label{section:complete_info_game}
In this section, we will formulate and analyze the \textbf{Joint Communication Association and Computation Offloading (JCACO) game} in the complete information scenario, denoted by $\G^{\mathrm{C0}}$.
In what follows, we first formulate the JCACO game $\G^{\mathrm{C0}}$, which is subsequentially decoupled into two subgames, i.e., a communication association game $\G^{\mathrm{C1}}$ and a computation offloading game $\G^{\mathrm{C2}}$. 
We then demonstrate that both $\G^{\mathrm{C1}}$ and $\G^{\mathrm{C2}}$ are potential games, thus guaranteeing  the existence of Nash equilibrium. 

\subsection{Game Formulation under Complete Information \label{subsection:complete_info_JCACO_game}}
The JCACO game under complete information, denoted as $\G^{\mathrm{C0}}$, can be defined as follows:

\begin{game}[JCACO Game under complete information]
	The JCACO game under complete information, can be defined as:
	\begin{equation}\small \label{game_C0}
		\begin{aligned}
			\G^{\mathrm{C0}}\triangleq \{ \N_A,\{\boldsymbol{s}_n\}_{n\in \N_A}, \{  T_{n}^{\mathrm{Total}}(\boldsymbol{s}_n,\mathbf{S}_{-n};\boldsymbol{\omega}) \}_{n\in \N_A} \}
		\end{aligned}
	\end{equation}
	where the key elements are defined as follows: 
	\begin{itemize}
		\item[$\bullet$] \textbf{Players:}
		All active UEs, indexed by $n \in \N_{A}$;
		
		\item[$\bullet$] \textbf{Strategies:}
		The strategy of each active UE $n$ is  $\boldsymbol{s}_n \triangleq (\boldsymbol{x}_n, \boldsymbol{y}_n, \boldsymbol{d}_n)$, representing the UE-AP association strategy, computation offloading strategy, and inference step selection strategy, respectively.
		
		\item[$\bullet$] \textbf{Payoffs:} The payoff for UE $n \in \N_{A}$ is the total service time, denoted by $T_{n}^{\mathrm{Total}}(\boldsymbol{s}_n,\mathbf{S}_{-n};\boldsymbol{\omega})$, which depends not only on its own strategy $\boldsymbol{s}_{n}$, but also the strategies of all other UEs except $n$,   i.e., $\mathbf{S}_{-n}=(\boldsymbol{s}_1,...,\boldsymbol{s}_{n-1},\boldsymbol{s}_{n+1},...,\boldsymbol{s}_{N})$.
	\end{itemize}
\end{game}
The formal definition of the NE for $\G^{\mathrm{C0}}$ is given as follows:
\begin{definition} [{NE Under Complete Information}]
	A strategy profile $\mathbf{S}^{*}=(\boldsymbol{s}_{1}^{*},...,\boldsymbol{s}_{N}^{*})$
	 an \textbf{NE} of $\G^{\mathrm{C0}}$, if for every active player $n \in \N_A$, the following condition holds:
	\begin{equation} \small \label{def_NE}
		\begin{aligned}
			T_{n}^{\mathrm{Total}}(\boldsymbol{s}^{*}_n,\mathbf{S}^{*}_{-n};\boldsymbol{\omega}) \leq T_{n}^{\mathrm{Total}}(\boldsymbol{s}_n,\mathbf{S}^{*}_{-n};\boldsymbol{\omega}), \forall n\in \N_{A},
		\end{aligned}
	\end{equation}
	where $\mathbf{S}_{-n}^{*}=(\boldsymbol{s}^{*}_{1},\dots,\boldsymbol{s}^{*}_{n-1},\boldsymbol{s}^{*}_{n+1}\dots,\boldsymbol{s}^{*}_{N})$.
\end{definition}
To simplify the analysis, we decouple $\G^{\mathrm{CO}}$ into two subgames: the communication associating game $\G^{\mathrm{C1}}$ and the computation offloading game $\G^{\mathrm{C2}}$. Formally, the communication associating game $\G^{\mathrm{C1}}$, is defined as follows:
\begin{game}[Communication Association Game under Complete Information]
	The communication association game under complete information, can be defined as:
	\begin{equation}\small \label{game_C1}
		\begin{aligned}
			\G^{\mathrm{C1}}\triangleq \{ \N_A, \{\boldsymbol{x}_n\}_{n\in \N_A}, \{T_{n}^{\mathrm{Acc}}(\boldsymbol{x}_n,\mathbf{X}_{-n};\boldsymbol{\omega}) \}_{n\in \N_A} \},
		\end{aligned}
	\end{equation}	
	where the key elements are defined as follows: 
	\begin{itemize}
		\item[$\bullet$] \textbf{Players:}
		All active UEs, $n \in \N_{A}$ in the system;
		
		\item[$\bullet$] \textbf{Strategies:}
		The strategy of each active UE $n$ is $\boldsymbol{x}_n$.
		
		\item[$\bullet$] \textbf{Payoffs:} The payoff for UE $n \in \N_A$ is its transmission completion time 		$T_{n}^{\mathrm{Acc}}(\boldsymbol{x}_n,\mathbf{X}_{-n};\boldsymbol{\omega})$, which depends  on its own strategy $\boldsymbol{x}_{n}$ and the strategies of all the other UEs, i.e., $\mathbf{X}_{-n}=(\boldsymbol{x}_1,...,\boldsymbol{x}_{n-1},\boldsymbol{x}_{n+1},...,\boldsymbol{x}_{N})$.
	\end{itemize}
\end{game}
The computation offloading game $\G^{\mathrm{C2}}$ under complete information, can be defined as follows:
\begin{game}[Computation Offloading Game under Complete Information]
	The Computation Offloading Game under Complete Information, can be defined as:
	\begin{equation} \small \label{game_C2}
		\begin{aligned}
			\G^{\mathrm{C2}}\triangleq \{ \N_A,\{\boldsymbol{z}_n\}_{n\in \N_A}, \{ T_{n}^{\mathrm{Comp}}(\boldsymbol{z}_n,\mathbf{Z}_{-n};\boldsymbol{\omega}) \}_{n\in \N_A} \},
		\end{aligned}
	\end{equation}
	where the key elements are defined as follows: 
	\begin{itemize}
		\item[$\bullet$] \textbf{Players:}
		All active UEs, $n \in \N_{A}$ in the system;
		
		\item[$\bullet$] \textbf{Strategies:}
		The strategy of each active UE $n$ is  $\boldsymbol{z}_{n}\triangleq (\boldsymbol{y}_{n},\boldsymbol{d}_{n})$, where $\boldsymbol{y}_n$ is the computation offloading strategy, $\boldsymbol{d}_{n}$ is the inference step strategy.
		
		\item[$\bullet$] \textbf{Payoffs:} The payoff for UE $n \in \N_A$ is its computation completion time, i.e.,
		$T_{n}^{\mathrm{Comp}}(\boldsymbol{z}_n,\mathbf{Z}_{-n};\boldsymbol{\omega})$, which depends on its strategy $\boldsymbol{z}_{n}$, and the strategies of all the other UEs, i.e., $\mathbf{Z}_{-n}=(\boldsymbol{z}_1,...,\boldsymbol{z}_{n-1},\boldsymbol{z}_{n+1},...,\boldsymbol{z}_{N})$.
	\end{itemize}
\end{game}
The NE of subgames $\G^{\mathrm{C1}}$ and $\G^{\mathrm{C2}}$ can be defined in the same way as that for $\G^{\mathrm{C0}}$. Due to space limit, we do not present the detailed definitions. 

\subsection{Analysis of Communication Association Game $\G^{\mathrm{C1}}$}
In this subsection, we will analyze the NE of $\G^{\mathrm{C1}}$ through the property of potential game \cite{potential_game}. 

\begin{definition}[{Potential Game}]
	A game is a potential game if and only if there exists a potential function  $\Phi(\mathbf{X})$ such that:
	\begin{equation} \small \label{potential_game_def}
		\begin{aligned}
			&T_{n}^{\mathrm{Acc}}(\boldsymbol{x}'_{n},\mathbf{X}_{-n}) < T_{n}^{\mathrm{Acc}}(\boldsymbol{x}_{n},\mathbf{X}_{-n}) \\
			&\Leftrightarrow
			\Phi(\boldsymbol{x}'_{n},\mathbf{X}_{-n}) < \Phi(\boldsymbol{x}_{n},\mathbf{X}_{-n}), 
			\forall n\in \N, \boldsymbol{x}_{n},\boldsymbol{x}'_{n} \in \mathbf{X}.
		\end{aligned}
	\end{equation}	
\end{definition}
Eq.(\ref{potential_game_def}) implies that, in a potential game, the change in the potential function $\Phi(\cdot)$, due to an unilateral strategy deviation, must share the same sign as the change in the payoff function.
This implies that any unilateral change by a player that improves its payoff also leads to an improvement in the potential function.
We then show that $\G^{\mathrm{C1}}$ is a potential game through the following theorem.

\begin{theorem}\label{theorem:UA}
	The game $\G^{\mathrm{C1}}$ is a potential game with the following potential function:
	\begin{equation} \small \label{potential_func_AP}
		\begin{aligned}
			\Phi(\mathbf{X};\boldsymbol{\omega}) =\sum_{m=1}^{M} \varphi^{L_{m}(\mathbf{X};\boldsymbol{\omega})},
		\end{aligned}
	\end{equation}
	where $\varphi$ is any value larger than a threshold $\varphi_0 \triangleq \sqrt[\epsilon]{2}$, where $\epsilon$ is the minimum time scale of communication system.	
\end{theorem}
\begin{IEEEproof}
	Please refer to Appendix A.
\end{IEEEproof}
\begin{lemma}
	The game $\G^{\mathrm{C1}}$ has at least one Nash equilibrium.
\end{lemma}
Intuitively, at the NE point $(\boldsymbol{x}_{n}^{*},\mathbf{X}_{-n}^{*})_{n\in \N_{A}}$, the potential
function cannot be further minimized. According to (\ref{potential_game_def}), no player can unilaterally reduce its service completion time, which implies that the corresponding strategy profile constitutes an NE.

\subsection{Analysis of Computation Offloading Game $\G^{\mathrm{C2}}$}
\begin{theorem}\label{theorem:CO}
	The game $\G^{\mathrm{C2}}$ is a potential game with the following potential function:
	\begin{equation} \small \label{potential_func_offload}
		\begin{aligned}
			\Psi(\mathbf{Y,D};\boldsymbol{\omega}) =\sum_{k=1}^{K}\phi^{I_{k}(\mathbf{Y,D};\boldsymbol{\omega})}, 			
		\end{aligned}
	\end{equation}
	where $\phi$ is any value larger than a threshold $\phi_0 \triangleq \sqrt[l]{2}$, where $l$ is the minimum time scale of computation system.
\end{theorem}
\begin{IEEEproof}
	Please refer to Appendix B.
\end{IEEEproof}

According to the finite improvement property (FIP) of potential games \cite{potential_game},
we can show that in both $\G^\mathrm{C1}$ and $\G^\mathrm{C2}$,  any best-response (BR) dynamics will terminate at an NE. 
The detailed BR algorithms for $\mathcal{G}^{\mathrm{C1}}$ and $\mathcal{G}^{\mathrm{C2}}$ have been studied in our previous work \cite{liu_wiopt2024}.

\section{Game Formulation and Analysis under Stochastic Information} \label{section:sto_info_game}
In this section, we analyze the JCACO game in the stochastic environment. 
We begin by presenting the game formulation under stochastic information and then proceed to analyze the properties of this game.

\subsection{Game Formulation under Stochastic Information}\label{subsection:stochastic_JCACO_game}
To better capture the behavior of UEs and the system performance under the dynamic activity of UEs, we reformulate the game $\G^{\mathrm{C0}}$ as a stochastic game, denoted as $\G^{\mathrm{S0}}$.
The key difference between $\G^{\mathrm{C0}}$ and $\G^{\mathrm{S0}}$, lies in the dynamics of the UEs' activeness.
In $\G^{\mathrm{C0}}$, strategies are determined under complete information about the system state. In contrast, in $\G^{\mathrm{S0}}$, each UE selects its strategy based on the dynamic activity patterns over the sample space $\Omega$. Formally, 

\begin{game}
	[JCACO game under stochastic information]	
	The JCACO game under stochastic information can be defined as:
	\begin{equation}\small 
		\begin{aligned}
			\G^{\mathrm{S0}}\triangleq \{ \N,\{\boldsymbol{s}_n\}_{n\in \N}, \{ \overline{T} _{n}^{\mathrm{Total}}(\boldsymbol{s}_n,\mathbf{S}_{-n}) \}_{n\in \N} \},
		\end{aligned}
	\end{equation}
where the key elements are defined as follows:
		\begin{itemize}
		\item[$\bullet$] \textbf{Players:}
		All UEs, indexed by $n \in \N$;
		
		\item[$\bullet$] \textbf{Strategies:}
		The strategy of each UE $n$ is  $\boldsymbol{s}_n \triangleq (\boldsymbol{x}_n, \boldsymbol{y}_n, \boldsymbol{d}_n)$, representing the communication association, computation offloading strategy, and inference step selection strategy, respectively.
		
		\item[$\bullet$] \textbf{Payoffs:} The expected payoff for UE $n \in \N$ is the expected total service time, denoted by $\overline{T} _{n}^{\mathrm{Total}}(\boldsymbol{s}_n,\mathbf{S}_{-n})$, which depends not only on its own strategy $\boldsymbol{s}_{n}$, but also the strategies of all other UEs except $n$,   i.e., $\mathbf{S}_{-n}=(\boldsymbol{s}_1,...,\boldsymbol{s}_{n-1},\boldsymbol{s}_{n+1},...,\boldsymbol{s}_{N})$.
	\end{itemize}	
\end{game}

The expected payoff function for UE $n \in \N$ is defined as:
	\begin{equation}\label{avg_obj} \small
		\begin{aligned}
			\overline{T} _{n}^{\mathrm{Total}}(\boldsymbol{s}_n,\mathbf{S}_{-n}) &\triangleq \mathbb{E}_{\boldsymbol{\omega}\in \Omega}[T_{n}^{\mathrm{Total}}(\boldsymbol{s},\mathbf{S}_{-n};\boldsymbol{\omega})]  \\
			&= \overline{T}_{n}^{\mathrm{Acc}}(\mathbf{X}) +  \overline{T}_{n}^{\mathrm{Comp}}(\mathbf{Y,D}), 
		\end{aligned}
	\end{equation}
which consists of the expected transmission   time $\overline{T}_{n}^{\mathrm{Acc}}(\mathbf{X})$ and the expected computation   time $\overline{T}_{n}^{\mathrm{Comp}}(\mathbf{Y,D})$: 
\begin{equation} \label{exp_acc} \small
	\begin{aligned}
		\overline{T}_{n}^{\mathrm{Acc}}(\mathbf{X}) &= \mathbb{E}_{\boldsymbol{\omega}\in \Omega}[T_{n}^{\mathrm{Acc}}(\mathbf{X};\boldsymbol{\omega})] \\ 
		&= \sum_{m\in \M} x_{nm} \cdot \overline{L}_m (\mathbf{X}) ,\forall n \in \N,
	\end{aligned}
\end{equation}
\begin{equation} \label{exp_comp} \small 
	\begin{aligned}
		\overline{T}_{n}^{\mathrm{Comp}}(\mathbf{Y,D}) &=  \mathbb{E}_{\boldsymbol{\omega}\in\Omega}[T_{n}^{\mathrm{Comp}}(\mathbf{Y,D};\boldsymbol{\omega}) ] \\
		&= \sum_{k\in \K} y_{nk} \cdot \overline{I}_{k}(\mathbf{Y,D}), \forall n\in \N.
	\end{aligned}
\end{equation}
In the   equation \eqref{exp_acc}, $\overline{L}_m (\mathbf{X})$ represents the expected load of AP $m$, which is given by:
\begin{equation} \label{def_exp_load_AP} \small
	\begin{aligned}
		\overline{L}_m (\mathbf{X}) = \mathbb{E}_{\boldsymbol{\omega}\in\Omega}[L_{m}(\mathbf{X};\boldsymbol{\omega})] 
		= \sum_{\boldsymbol{\omega}\in \Omega} p({\boldsymbol{\omega}})\cdot L_{m}(\mathbf{X};\boldsymbol{\omega}), \forall m\in \M. \\
	\end{aligned}
\end{equation}
In the   equation \eqref{exp_comp}, $\overline{I}_{k}(\mathbf{Y,D})$ represents the expected load of ES $k$, which is  given by:
\begin{equation} \label{exp_load_ES} \small
	\begin{aligned}
		\overline{I}_{k}(\mathbf{Y,D})&=\mathbb{E}_{\boldsymbol{\omega}\in\Omega}[I_{k}(\mathbf{Y,D};\boldsymbol{\omega})]\\
		&=\sum_{\boldsymbol{\omega}\in \Omega} p({\boldsymbol{\omega}})\cdot
		I_k(\mathbf{Y,D};\boldsymbol{\omega}), \forall k \in \K.
	\end{aligned}
\end{equation}

Similar as in Section \ref{subsection:stochastic_JCACO_game}, $\G^{\mathrm{S0}}$ can be decoupled into two subgames: $\G^{\mathrm{S1}}$ and $\G^{\mathrm{S2}}$, defined below. 

\begin{game}[Stochastic Communication Association Game]
	The communication association game under stochastic information of UE activeness dynamics, is defined as: 
	\begin{equation} \small
		\begin{aligned}
			\G^{\mathrm{S1}}\triangleq \{ \N,\{\boldsymbol{x}_n\}_{n\in \N}, \{\overline{T}_{n}^{\mathrm{Acc}}\{\boldsymbol{x}_n,\mathbf{X}_{-n} \}_{n\in \N} \},
		\end{aligned}
	\end{equation}
where the key elements are defined as follows:
	\begin{itemize}
		\item[$\bullet$] \textbf{Players:}
		All UEs, indexed by $n \in \N$;
		
		\item[$\bullet$] \textbf{Strategies:}
		The strategy of each active UE $n$ is $\boldsymbol{x}_{n}$.
		
		\item[$\bullet$] \textbf{Payoffs:} The payoff for UE $n \in \N$ is its expected transmission completion time, i.e.,
		$\overline{T}_{n}^{\mathrm{Acc}}(\boldsymbol{x}_n,\mathbf{X}_{-n})$, which depends on its own strategy $\boldsymbol{x}_{n}$, and the strategies of all the other UEs, i.e., $\mathbf{X}_{-n}=(\boldsymbol{x}_1,...,\boldsymbol{x}_{n-1},\boldsymbol{x}_{n+1},...,\boldsymbol{x}_{N})$.
	\end{itemize}
	
\end{game}

\begin{game}[Stochastic Computation Offloading Game]
	The Computation Offloading Game, under stochastic information of UE activeness dynamics, is defined as:
	\begin{equation} \small
		\begin{aligned}
			\G^{\mathrm{S2}}\triangleq \{ \N,\{\boldsymbol{z}_n\}_{n\in \N}, \{ \overline{T}_{n}^{\mathrm{Comp}}(\boldsymbol{z}_n,\mathbf{Z}_{-n}) \}_{n\in \N} \}, 
		\end{aligned}
	\end{equation}
where the key elements are defined as follows:
	\begin{itemize}
		\item[$\bullet$] \textbf{Players:}
		All UEs, indexed by $n \in \N$;
		
		\item[$\bullet$] \textbf{Strategies:}
		The strategy of each active UE $n$ is $\boldsymbol{z}_{n} \triangleq (\boldsymbol{y}_{n}, \boldsymbol{d}_{n})$, where $\boldsymbol{y}_{n}$ is the computation offloading strategy, $\boldsymbol{d}_{n}$ is the inference step strategy.
		
		\item[$\bullet$] \textbf{Payoffs:} The payoff for UE $n \in \N$ is its expected computation completion time, i.e.,
		$\overline{T}_{n}^{\mathrm{Comp}}(\boldsymbol{z}_n,\mathbf{Z}_{-n})$, which depends on its own strategy $\boldsymbol{z}_{n}$, and the strategies of all the other UEs, i.e., $\mathbf{Z}_{-n}=(\boldsymbol{z}_1,...,\boldsymbol{z}_{n-1},\boldsymbol{z}_{n+1},...,\boldsymbol{z}_{N})$.
	\end{itemize}
\end{game}

\subsection{Analysis of Stochastic Communication Association Game $\G^{\mathrm{S1}}$}
Now we analyze the NE of stochastic communication associating game $\G^{\mathrm{S1}}$ through the property of potential game.
\begin{theorem}
	The game $\G^{\mathrm{S1}}$ is a stochastic potential game with the following potential function:
	\begin{align}\label{exp_potential_func_trans}
		\overline{\Phi}(\mathbf{X}) = \sum_{m=1}^{M}\tilde{\varphi}^{\overline{L}_{m}(\mathbf{X})},
	\end{align}
	where $\tilde{\varphi}$ is any value larger than a threshold $\tilde{\varphi}=\sqrt[\epsilon]{2}$, where $\epsilon$ is the minimum time scale of communication system.
\end{theorem}

\begin{IEEEproof}
	Please refer to Appendix C.
\end{IEEEproof}

\subsection{Analysis of Stochastic Computation Offloading Game $\G^{\mathrm{S2}}$}
Now  we analyze the stochastic computation offloading game $\G^{\mathrm{S2}}$ through the property of potential game. 

\begin{theorem}
	The game $\G^{\mathrm{S2}}$ is a stochastic potential game with the following expected potential function $\overline{\Psi}(\mathbf{Y,D})$:
	\begin{equation} \small \label{exp_potential_func_offload}
		\begin{aligned}
			\overline{\Psi}(\mathbf{Y,D}) =\sum_{k=1}^{K}\tilde{\phi}^{\overline{I}_{k}(\mathbf{Y,D})}, 			
		\end{aligned}
	\end{equation}
	where $\tilde{\phi}$ is any value larger than a threshold $\tilde{\phi_0}=\sqrt[l]{2}$, where $l$ is the minimum time scale of computation system.
\end{theorem}
\begin{IEEEproof}
	Please refer to Appendix D.
\end{IEEEproof}	

It is well known that in potential games, simple algorithms such as best-response \cite{potential_game} are guaranteed to converge to a pure-strategy NE.
However, these algorithms require complete information about  the strategies of all other players at each iteration. 
Moreover, they assume a static environment, i.e., fixed player set, time-invariant payoff functions, and no stochastic system dynamics during the convergence process. 
As a result, classic game-theoretic algorithms are likely to fail to identify the NE in stochastic games such as $\mathcal{G}^{\mathrm{S0}}$. 
This   motivates us design an efficient algorithm for reaching the NE of $\mathcal{G}^{\mathrm{S0}}$.

\section{Multi-Agent Stochastic Learning Algorithm Design} \label{section:algorithm}
In this section, we design a fully decentralized \textbf{Multi-Agent Stochastic Learning (MASL)} framework for finding NE in the stochastic games $\mathcal{G}^{\mathrm{S1}}$ and $\mathcal{G}^{\mathrm{S2}}$.
We begin by introducing the core idea of MASL, followed by detailed algorithmic implementations, and finally discuss convergence properties and computational complexity.

\subsection{Core Idea of MASL}
The MASL framework is a decentralized learning paradigm where agents autonomously optimize their decisions through repeated interactions with the environment. 
Each agent maintains a \emph{mixed-strategy}, which is a  probability distribution defined over its set of available actions. 
This mixed-strategy is iteratively refined over time. 
In each iteration, the agent selects an action based on its current strategy, receives stochastic reward feedback from the environment, and updates the strategy in response to the observed outcome.
This cycle of exploration (via probabilistic action selection) and exploitation (via reward-driven updates) enables agents to adaptively improve their strategies. 
Over successive iterations, the behavior of all agents converges toward a stable probability distribution, typically a Nash equilibrium \cite{SLA0,SLA_textbook}.

Let	$\boldsymbol{p}^{\tau}_{n}=(p^{\tau}_{n1},...,p^{\tau}_{n|\mathcal{A}_{n}|})$ denote the mixed strategy of agent $n$ at iteration $\tau$, where $p^{\tau}_{na}$ is the probability of selecting action $a\in \mathcal{A}_{n}$.
To begin with unbiased exploration, all actions are initially assigned equal probability:
\begin{equation}\small \label{unified_initialization}
	\begin{aligned}
		\boldsymbol{p}_n^0 = \left( \frac{1}{|\mathcal{A}_{n}|}, \dots, \frac{1}{|\mathcal{A}_{n}|} \right)_{1\times |\mathcal{A}_{n}|}, \quad \forall n \in \mathcal{N}.
	\end{aligned}
\end{equation}
In each iteration $\tau = 1, 2, \dots$, each agent performs two steps. 
First, it samples an action $a^{\tau}_{n}$ according to its current strategy vector
$\boldsymbol{p}_n^\tau$, and receives a corresponding reward $r_{n}^\tau\in [0,1]$.
Second, the agent updates its strategy vector based on the received reward according to the following rule:  
\begin{equation} \small \label{unified_update_rule}
	p_{na}^{\tau+1} = 
	\begin{cases}
		p_{na}^{\tau} + \eta r_n^\tau (1 - p_{na}^{\tau}), & \text{if action } a \text{ is selected}, \\
		p_{na}^{\tau} - \eta r_n^\tau p_{na}^{\tau}, & \text{otherwise},
	\end{cases}
\end{equation}
where $\eta \in (0,1)$  denotes the learning rate. 
Intuitively, \emph{this update rule increases the probability of the selected action if the reward is positive, while proportionally decreasing the probabilities of   other actions}. 
Meanwhile, this update rule preserves the normalization condition $\sum_{a=1}^{|\mathcal{A}_n|}p^{\tau}_{na}=1$ and serves as the mathematical foundation of the MASL dynamics analyzed in the following subsection. 

\subsection{MASL Algorithms for Subgames $\mathcal{G}^{\mathrm{S1}}$ and $\mathcal{G}^{\mathrm{S2}}$}
Building upon the unified initialization mechanism in \eqref{unified_initialization} and  the strategy 
update rule in \eqref{unified_update_rule}, we design MASL algorithms tailored for the communication association game $\mathcal{G}^{\mathrm{S1}}$ 
and the computation offloading game $\mathcal{G}^{\mathrm{S2}}$. 
They differ in terms of action spaces and reward definitions to accommodate the distinct objectives of the two subgames.
\subsubsection{MASL-Algorithm for $\G^{\mathrm{S1}}$}

\textcolor{black}{
	In $\G^{\mathrm{S1}}$, the action space of UE $n$ is the set of APs $\mathcal{M}$.
	At each iteration, each active UE $n \in \mathcal{N}_A$ selects an AP according to its mixed strategy $\mathbf{p}_n^\tau$, observes its expected transmission delay $\overline{T}_n^{\mathrm{Acc}(\tau)}$ via~\eqref{exp_acc}, and then updates its strategy based on the following normalized reward:
	\begin{equation}\small \label{reward_AP}
		\tilde{r}_n^\tau = 1 - v_n \cdot \overline{T}_n^{\mathrm{Acc}(\tau)},
	\end{equation}
	where $v_n \triangleq 1 / \widehat{T}_n^{\mathrm{Acc}}$ is the normalization factor.
	Here, $\widehat{T}_{n}^{\mathrm{Acc}} > 0$ denotes a fixed upper bound on the communication delay of UE $n$.
	It is chosen such that
	$\overline{T}_{n}^{\mathrm{Acc}(\tau)} \le \widehat{T}_{n}^{\mathrm{Acc}}$ over the considered learning horizon, so that the normalized reward remains bounded within $(0,1]$.
	It is worth noting that the normalization factor is a fixed constant and does not depend on future observations or online updates.
	Therefore, at each iteration, the reward is computed only from the currently observed delay and the pre-specified normalization constant.
}
	The mixed strategy $\mathbf{p}^{\tau}_{n}$ is then updated as follows:
	\begin{equation}\small \label{update_rule_UA}
		p_{nm}^{\tau+1} = 
		\begin{cases}
			p_{nm}^{\tau} + \alpha \tilde{r}_n^\tau (1 - p_{nm}^{\tau}), & \text{if AP } m \text{ is selected}, \\
			p_{nm}^{\tau} - \alpha \tilde{r}_n^\tau p_{nm}^{\tau}, & \text{otherwise},
		\end{cases}
	\end{equation}
	where $\alpha \in (0,1)$ is the learning rate.
	Inactive UEs retain their strategies from the previous iteration.
	The MASL algorithm for the communication association game $\mathcal{G}^{\mathrm{S1}}$ is summarized in Alg.~\ref{alg:MASL_for_UA}.

\begin{algorithm}[t]
	\renewcommand{\algorithmicrequire}{\textbf{Input:}}
	\renewcommand{\algorithmicensure}{\textbf{Output:}}
	\caption{MASL algorithm for communication association $\G^{\mathrm{S1}}$}
	\label{alg:MASL_for_UA}
	\begin{algorithmic}[1]
		\REQUIRE Network parameters, learing rate $\alpha$
		\ENSURE	Equilibrium communication association strategies of all UEs $n\in \N$
		\STATE \textbf{Initialize} $\tau=0$, $\mathbf{p}_{n}^{0}=(\frac{1}{M},\dots,\frac{1}{M}), \forall n\in \N$
		\WHILE{not converged}
		\FOR{each $n \in \N_{A}$}		
		\STATE Sample a communication association strategy (i.e., an AP) according to $\mathbf{p}^{\tau}_{n}$
		\STATE Experience and record $\overline{T}^{\mathrm{Acc}(\tau)}_{n}$
		\STATE Compute reward $\tilde{r}^{\tau}_{n}$ via \eqref{reward_AP}
		\STATE Update $\mathbf{p}^{\tau+1}_{n}$ via \eqref{update_rule_UA}
		\ENDFOR
		\STATE Inactive UEs $ n\in (\N \setminus \N_{A})$ keep $\mathbf{p}^{\mathbf{\tau}}_{n}$
		\IF{ $\|\mathbf{p}_n^{\tau+1} - \mathbf{p}_n^\tau\|_2 < \delta, \forall n\in \N$  }
		\STATE \textbf{Break}
		\ENDIF
		\STATE $\tau \Leftarrow \tau+1$  	
		\ENDWHILE
	\end{algorithmic}
	\vspace{-1mm}
\end{algorithm}

\subsubsection{MASL-Algorithm for $\G^{\mathrm{S2}}$}
\begin{algorithm}[t]
	\renewcommand{\algorithmicrequire}{\textbf{Input:}}
	\renewcommand{\algorithmicensure}{\textbf{Output:}}
	\caption{MASL algorithm for computation offloading $\G^{\mathrm{S2}}$}
	\label{alg:MASL_for_CO}
	\begin{algorithmic}[1]
		\REQUIRE Network parameters, learning rate $\beta$
		\ENSURE Equilibrium computation offloading  strategies of all UEs $n\in \N$ 
		\STATE Compute $d_{nk}^{*}$ via \eqref{optimal_IS} $\forall n,k$ 	
		\STATE \textbf{Initialize} $\tau=0$, $\mathbf{q}_{n}^{0}=(\frac{1}{K},\dots,\frac{1}{K}), \forall n\in \N$
		\WHILE{not converged}
		\FOR{each $n \in \N_{A}$}	
		\STATE Sample a computation offloading strategy (i.e., an ES) according to $\mathbf{q}^{\tau}_{n}$
		\STATE Experience and record  $\overline{T}^{\mathrm{Comp}(\tau)}_{n}$ 
		\STATE Compute reward $\tilde{\mathtt{r}}^{\tau}_{n}$ via \eqref{reward_ES}
		\STATE Update $\mathbf{q}^{\tau+1}_{n}$ via \eqref{update_rule_CO}  		
		\ENDFOR \\
		\STATE Inactive UEs $ n\in (\N \setminus \N_{A})$ keep $\mathbf{q}^{\tau}_{n}$.		
		\IF{ $\|\mathbf{q}_n^{\tau+1} - \mathbf{q}_n^\tau\|_2 < \delta, \forall n\in \N$  }
		\STATE \textbf{Break}
		\ENDIF \\		
		\STATE $\tau \Leftarrow \tau+1$ \\
		\ENDWHILE
	\end{algorithmic}
	\vspace{-1mm}
\end{algorithm}
In $\mathcal{G}^{\mathrm{S2}}$, before selecting an ES, each UE computes its optimal number of inference steps $d^{*}_{nk}$ to satisfy error constraints berfore the learning process: 
\begin{equation} \label{optimal_IS}
	\begin{aligned}
		d_{nk}^{*} = \min \bigg\lceil -\frac{1}{\gamma_{nk}} \cdot \ln \bigg(\frac{\E^0_n}{\epsilon_k^{\mathrm{fwd}}}\bigg) \bigg\rceil, \forall n\in \N, k \in \K.
	\end{aligned}
\end{equation}
\textcolor{black}{Then, at each iteration, each active UE samples an ES according to its mixed strategy, observes its expected computation delay $\overline{T}_n^{\mathrm{Comp}(\tau)}$ via~\eqref{exp_comp}, and updates its strategy based on the following normalized reward:
	\begin{equation}\small \label{reward_ES}
		\tilde{\mathtt{r}}_n^\tau = 1 - u_n \cdot \overline{T}_n^{\mathrm{Comp}(\tau)},
	\end{equation}
	where $u_n \triangleq 1 / \widehat{T}_n^{\mathrm{Comp}}$ is the normalization factor.
	Here, $\widehat{T}_{n}^{\mathrm{Comp}} > 0$ denotes a fixed upper bound on the computation delay of UE $n$.
	Similar as \eqref{update_rule_UA}, 
	$\overline{T}_{n}^{\mathrm{Comp}(\tau)} \le \widehat{T}_{n}^{\mathrm{Comp}}$ over the considered learning horizon, so that the normalized reward remains bounded within $(0,1]$.}
	The mixed strategy is then updated as:
	\begin{equation}\small \label{update_rule_CO}
		q_{nk}^{\tau+1} = 
		\begin{cases}
			q_{nk}^{\tau} + \beta \tilde{\mathtt{r}}_n^\tau (1 - q_{nk}^{\tau}), & \text{if ES } k \text{ is selected}, \\
			q_{nk}^{\tau} - \beta \tilde{\mathtt{r}}_n^\tau q_{nk}^{\tau}, & \text{otherwise},
		\end{cases}
	\end{equation}
	where $\beta \in (0,1)$ is the learning rate.
	Inactive UEs retain their strategies from the previous iteration.
	The MASL algorithm for the   game $\mathcal{G}^{\mathrm{S2}}$ is summarized in Alg.~\ref{alg:MASL_for_CO}.

It's noteworthy that the MASL algorithms for $\G^{\mathrm{S1}}$ and $\G^{\mathrm{S2}}$ are fully distributed, making them particularly suitable for large-scale, decentralized systems.
In both algorithms, the strategy updates depend solely on each UE's local action-reward feedback, and there is no need for any global knowledge of other UEs' strategies or states. 
This decentralized property allows MASL algorithms to operate efficiently in environments with non-negligible communication overhead, making them particularly suitable for large-scale systems.

\subsection{Convergence of MASL Algorithms}
To analyze the learning dynamics of Alg.~\ref{alg:MASL_for_UA} and Alg.~\ref{alg:MASL_for_CO}, 
we model the strategy evolution of the UEs using ordinary differential equations (ODEs) as   \cite{SLA_textbook}. This allows for a mathematical representation of how the action selection probabilities evolve overtime. 
In this subsection, we examine the connection between the stationary points of the ODE and the NE points of the stochastic games.
Through this analysis, we show that for sufficiently small learning rates, the MASL algorithm converges to a pure-strategy NE.
This result underlines the stability and reliability of the MASL algorithm under appropriate conditions.\footnote{Due to the procedural similarity between Alg.\ref{alg:MASL_for_UA} and Alg.\ref{alg:MASL_for_CO}, the analysis framework developed here is directly applicable to both algorithms. For conciseness, we focus our discussion on Alg. \ref{alg:MASL_for_UA}.}

In the subsequent analysis, we model the discrete-time learning dynamics as a continuous-time dynamical system and explore the relationship between the learning dynamics and the NE points of $\G^{\mathrm{S1}}$.
\begin{lemma} \label{lemma:converge_ODE}
	Let $\mathbf{P} \triangleq (\mathbf{p}_1,\dots,\mathbf{p}_N)$ denote the mixed-strategy profile of the game $\G^{\mathrm{S1}}$, with initial state $\mathbf{P}^0 = [\frac{1}{M}]_{N\times M}$. For a sufficiently small learning rate ($\alpha \rightarrow 0$), the strategy profile ${ \mathbf{P}^{\tau} }$ generated by the learning update converges weakly to the solution trajectory of the following ODE:
	\begin{equation} \small \label{ODE-xx}
		\frac{d\mathbf{P}}{d\tau} = \mathscr{H}(\mathbf{P}), \qquad \mathbf{P}(0) = \mathbf{P}^{0},
	\end{equation}
	where $\mathscr{H}(\mathbf{P})$ is the expected update direction, defined as the conditional expectation:
	\begin{equation} \small \label{expect_condition_func}
		\mathscr{H}(\mathbf{P}) = \mathbb{E}\left[\mathscr{F}\bigl(\mathbf{P}^{\tau},\mathbf{X}^{\tau},\tilde{\mathbf{r}}^{\tau}\bigr) \big| \mathbf{P}^{\tau} = \mathbf{P}\right].
	\end{equation}
	Here, $\mathbf{X}^{\tau}=(\boldsymbol{x}_1^{\tau},\dots,\boldsymbol{x}_N^{\tau})$ is the action matrix at iteration $\tau$, $\tilde{\mathbf{r}}^{\tau}=(\tilde{r}_1^{\tau},\dots,\tilde{r}_N^{\tau})$ is the corresponding reward vector, and the function $\mathscr{F}(\cdot)$ encodes the update rule given in \eqref{unified_update_rule}, i.e., $\mathscr{F}(\mathbf{P}^{\tau},\mathbf{X}^{\tau},\tilde{\mathbf{r}}^{\tau}) = \mathbf{P}^{\tau+1}$.
\end{lemma}
\begin{IEEEproof}
	Please refer to Appendix E.
\end{IEEEproof}

Lemma \ref{lemma:converge_ODE} indicates that, the long-term behavior of the stochastic learning process can be precisely captured by an ODE, provided the agents’ learning rate is sufficiently small. 
Crucially, the lemma highlights that a sufficiently small step-size is a prerequisite for convergence. This theoretical requirement directly informs our parameter design and is empirically validated in the simulations of the next section, where we compare the system performance achieved under different learning rates.

Next, we will show that the strategy updating rule given in \eqref{unified_update_rule} leads to a solution (stationary point) of the corresponding ODE, which is exactly an NE of $\G^{\mathrm{S1}}$. Formally,

\begin{lemma} \label{lemma:small_learning_rate}
	With a sufficiently small step-size of learning rate  $\alpha\rightarrow 0$, Alg.1 converges to a stable stationary point of the ODE given in (\ref{ODE-xx}).
\end{lemma}
\begin{IEEEproof}
	Please refer to Appendix F.
\end{IEEEproof}

\begin{lemma} \label{lemma:NE_and_stationary}
	The following two statements are both true, and they are equivalent:
	\begin{enumerate}
		\item All the stable stationary points of the ODE in (\ref{ODE-xx}) are NE points of $\G^{\mathrm{S1}}$.
		\item All the NE points of $\G^{\mathrm{S1}}$ are stable stationary points of the ODE in (\ref{ODE-xx}).
	\end{enumerate}
\end{lemma}
\begin{IEEEproof}
	Please refer to Appendix G.
\end{IEEEproof}
Lemma \ref{lemma:NE_and_stationary} establishes a formal equivalence between the stable attractors of the learning dynamics (i.e., the stationary points of the ODE) and the NE of $\mathcal{G}^{\mathrm{S1}}$. 
This implies that the algorithm does not simply converge to a stationary point, but also reaches a strategically stable configuration where no agent has an incentive to deviate unilaterally.

\subsection{Complexity Analysis}
\label{subsec:complexity}

We analyze the computational and communication complexity of the MASL algorithms for both subgames.

\textbf{Time complexity.} In each iteration, an active UE performs three operations: (i) samples an action according to its probability vector ($\mathcal{O}(1)$), (ii) observes the corresponding delay and computes the normalized reward ($\mathcal{O}(1)$), and (iii) updates all entries in its probability vector. The update step dominates the per-iteration cost, yielding $\mathcal{O}(M)$ complexity per UE for $\mathcal{G}^{\mathrm{S1}}$ and $\mathcal{O}(K)$ for $\mathcal{G}^{\mathrm{S2}}$. The aggregate system complexity is therefore $\mathcal{O}(|\mathcal{N}_A| \cdot \max\{M, K\})$ per iteration.

\textbf{Space complexity.} Each UE maintains only local state: a probability vector of size $M$ (for $\mathcal{G}^{\mathrm{S1}}$) or $K$ (for $\mathcal{G}^{\mathrm{S2}}$), and for $\mathcal{G}^{\mathrm{S2}}$ additionally the precomputed optimal inference steps $\{d_{nk}^*\}_{k \in \mathcal{K}}$ requiring $\mathcal{O}(K)$ storage. Hence, the per-UE memory requirement is $\mathcal{O}(M)$ for $\mathcal{G}^{\mathrm{S1}}$ and $\mathcal{O}(K)$ for $\mathcal{G}^{\mathrm{S2}}$, independent of $N$.

\textbf{Communication overhead.} MASL requires no explicit exchange of strategies, channel states, or performance metrics among UEs. Delay measurements for reward computation are obtained through existing control-plane mechanisms (e.g., ACK/RTT feedback), incurring no algorithm-specific communication overhead. Convergence can be detected locally via $\|\mathbf{p}_n^{\tau+1} - \mathbf{p}_n^\tau\|_1 < \epsilon$, eliminating the need for global coordination. Consequently, the per-UE communication complexity is $\mathcal{O}(1)$ per iteration.

\section{Simulation Results} \label{section:simulations}
\begin{table}[t]
	\caption{Main Simulation Parameters}
	\label{table:parameters}
	\vspace{-1.0em}
	\begin{tabular}{lll}
		
		\toprule
		
		\textbf{Parameters} & \textbf{Value} \\
		
		\midrule
		
		Number of APs $M$ & [2,10] &    \\
		
		Number of ES $K$ & [2,10] &    \\
		
		Number of UEs $N$ & [10,30] &    \\
		
		UE's activeness probability $p_n$ & $ (0, 1] $ & \\
		
		FLOPS used by one inference step  $\xi_k$   & [0.1,0.5] TeraFLOPs/step & \\
		
		Computation capacity of ES $f_k$ & [2,10] TeraFLOPs/s &		\\
		
		UE transmission data volume $D_{n}$ & [2,10] MB &	\\
		
		Wireless bandwidth of AP $W$ & [2,10] MHz &	\\
		
		Transmission power of each UE $\rho_{nm}$ & 0.2 W & \\
		
		Path loss factor $\theta$ & 4 & \\
		
		The back-ground Noise Power $N_0$ & -174 dBm/Hz & \\
		
		\bottomrule
		\vspace{-2em}
	\end{tabular}
\end{table}

In this section, we illustate the performance of the proposed MASL algorithms for $\mathcal{G}^{\mathrm{S1}}$ and $\mathcal{G}^{\mathrm{S2}}$ through extensive simulations.
We investigate the performance of the algorithms across a variety of scenarios to demonstrate their effectiveness.
Specifically, we consider a $1km \times 1km$ square area, where APs are distributed evenly in the area.
Meanwile, the UEs follow a random walk model, where in each iteration round, each UE randomly moves to a location based on a specific probability distribution.
The other simulation parameters are summarized in Table \ref{table:parameters}. 
Similar experimental configurations are adopted in  \cite{xu_twc2025}, \cite{zhuang_infocom2025}, and \cite{liu_wiopt2024}.
Notably, we test various parameter settings across different scenarios to ensure the robustness and reliability of our results.

\begin{figure}[t]
	\centering
	\includegraphics[width=0.8\linewidth]{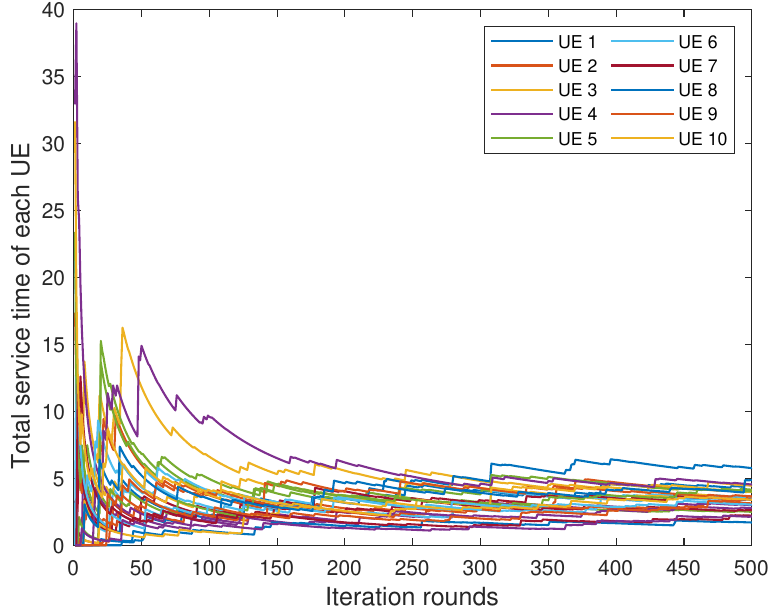}
	\vspace{-3mm}
	\caption{Dynamics of UE's service completion time.}
\label{fig:service_time_of_UE}
\vspace{-3mm}
\end{figure}

\begin{figure}[t]
	\centering
	\includegraphics[width=0.8\linewidth]{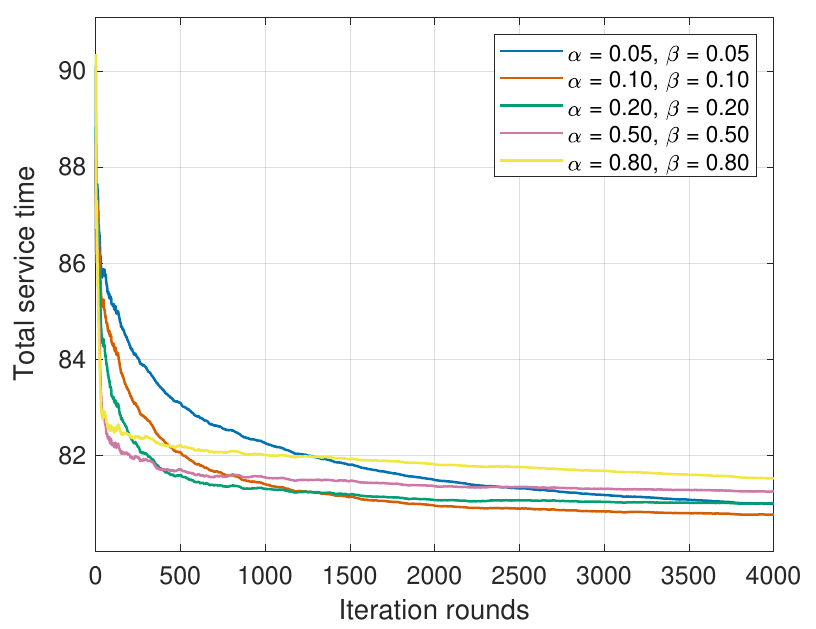}
	\vspace{-3mm}
	\caption{Impact of different of learning rate settings on MASL convergence.}
\label{fig:learning_rate_impact}
\vspace{-3mm}
\end{figure}

\begin{figure}[t]
	\centering
	\includegraphics[width=0.8\linewidth]{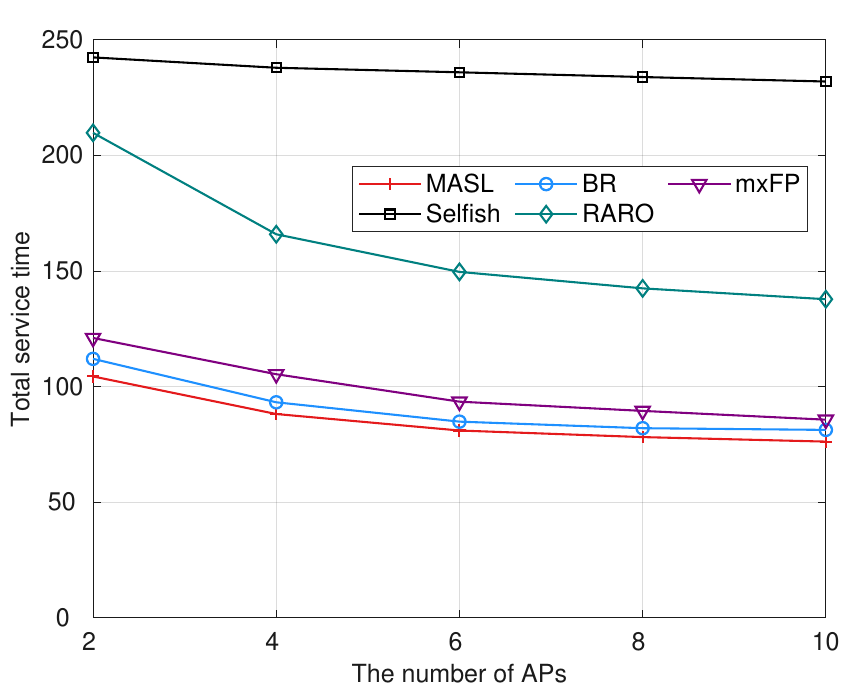}
	\vspace{-3mm}
	\caption{Total service time vs. the number of APs.}
\label{fig:Num_AP}
\vspace{-3mm}
\end{figure}

\begin{figure}[t]
	\centering
	\includegraphics[width=0.8\linewidth]{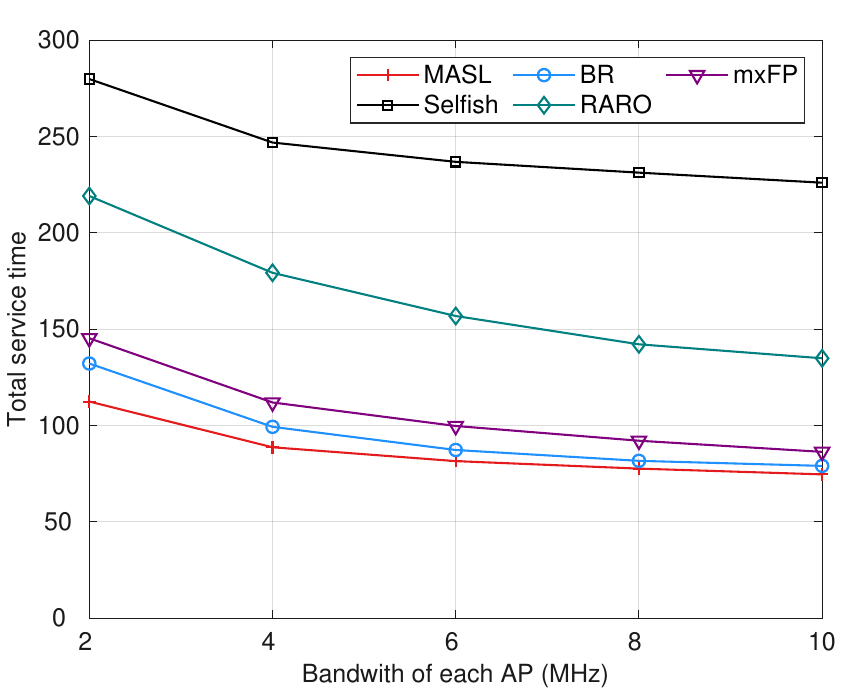}
	\vspace{-3mm}
	\caption{Total service time vs. wireless bandwidth of APs.}
	\label{fig:BW_of_AP}
	\vspace{-3mm}
\end{figure}

\begin{figure}[t]
	\centering
	\includegraphics[width=0.8\linewidth]{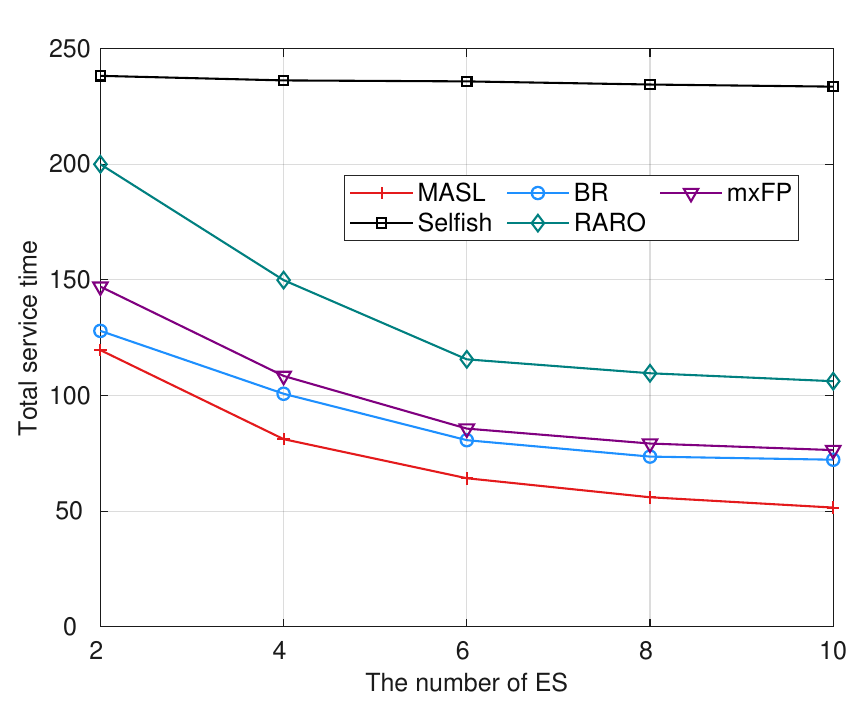}
	\vspace{-3mm}
	\caption{Total service time vs. the number of ES.}
	\label{fig:Num_ES}
	\vspace{-3mm}
\end{figure}

\begin{figure}[t]
	\centering
	\includegraphics[width=0.8\linewidth]{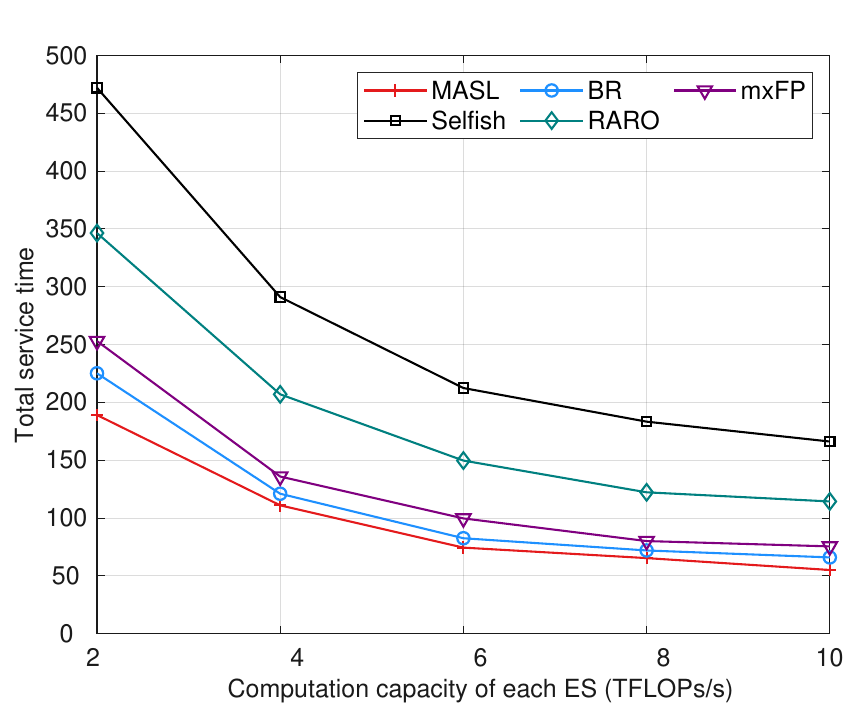}
	\vspace{-3mm}
	\caption{Total service time vs. computation capacity of ES.}
	\label{fig:comp_cap_of_ES}
	\vspace{-3mm}
\end{figure}

\begin{figure}[t]
	\centering
	\includegraphics[width=0.8\linewidth]{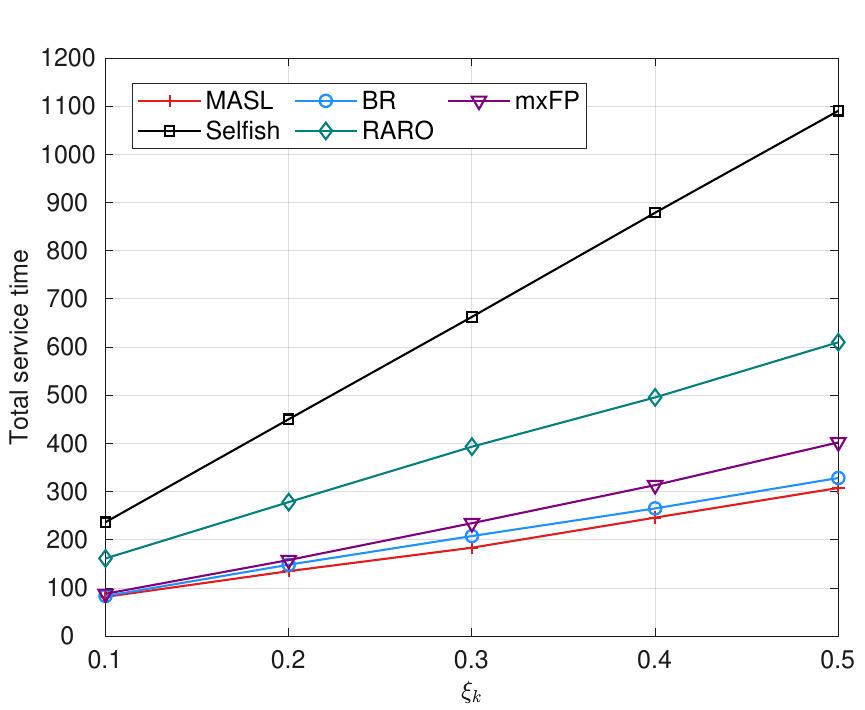}
	\vspace{-3mm}
	\caption{Total service time vs. TFLOPs per inference step}
	\label{fig:TFLOPs_each_IFstep}
	\vspace{-3mm}
\end{figure}

\begin{figure}[t]
	\centering
	\includegraphics[width=0.8\linewidth]{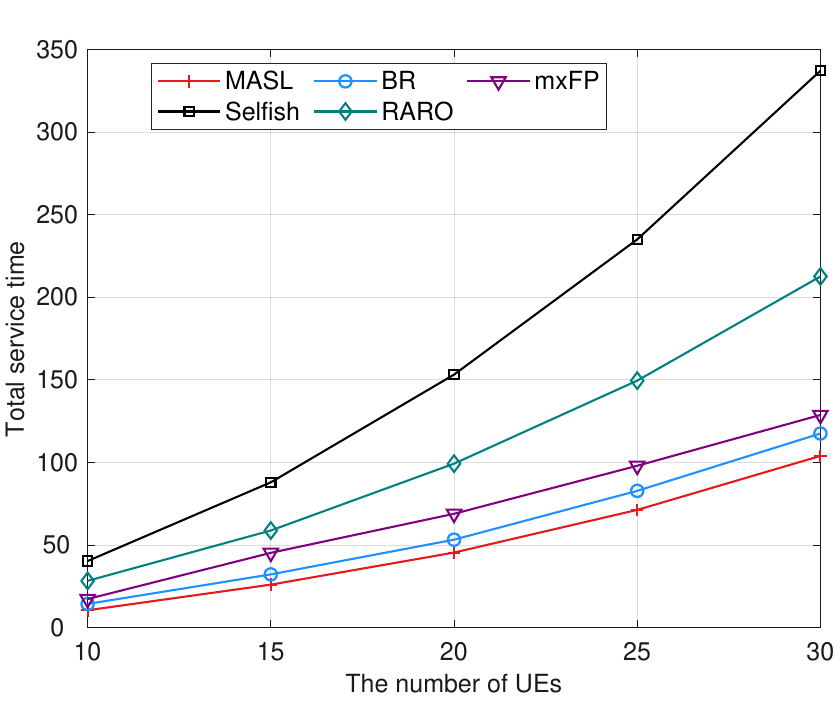}
	\vspace{-3mm}
	\caption{Total service time vs. the number of UEs.}
	\label{fig:Num_UE}
	\vspace{-3mm}
\end{figure}

\subsection{Benchmark Solutions}
To highlight the superiority of the MASL algorithms, we compare them against several state-of-the-art solutions. This comparison allows us to assess the relative performance of the proposed approach across different strategies and environments.
The following benchmark algorithms are used for comparison:
\begin{itemize}
	\item \textit{Best Response (BR)} \cite{game_theory}: 
	In the BR algorithm, each active UE adjusts its strategy to optimize its payoff, considering the strategies of other UEs. The BR algorithm iterates this process, with UEs updating their strategies until no UE can improve its payoff by unilaterally changing its strategy. A key assumption in the BR algorithm is that each UE has complete information about the strategies of the other UEs.	

	\item \textit{Fictitious Play with Mixed Strategy Estimation (mxFP)} \cite{fictitious_play}:
	Each UE updates its strategy by responding to the estimated mixed strategies of the other UEs, assuming these strategies are uniformly distributed across the available actions.

	\item \textit{Selfish UEs (Selfish)} \cite{greedy1}: In this approach, each UE independently selects the nearest AP for data transmission and chooses the ES with the largest computation capacity for task processing, ignoring the strategies of other UEs.
	
	\item \textit{Random UE-AP Association and Random computation Offloading (RARO)}: In this algorithm, UEs make random associations with APs and randomly select an ES, without considering any system information.
\end{itemize}

The primary distinction among the four benchmark algorithms lies in their information utilization and the type of strategic decision-making required. This difference underpins the varying performance of these solutions.
In the MASL algorithms, each UE updates its strategy based on local feedback (i.e., the reward signal) alone. In contrast, the BR algorithm requires UEs to know the number of active UEs and the strategies chosen by those UEs. The mxFP algorithm only requires knowledge of the total number of UEs in the system. The RARO algorithm makes decisions without any system information, while the Selfish algorithm relies on the nearest AP and the ES with the largest computational capacity.

Since the strategies of UEs are significantly influenced by APs, ES, and the stochastic activeness of the UEs, we conduct three sets of numerical experiments to examine how the total service time of UEs varies with each of these factors.

\subsection{Convergence}
We first investigate the dynamic evolution of service completion times for each UE over multiple iterations. With a total of 30 UEs, we observe that the service completion time initially decreases and then gradually stabilizes, as shown in Fig.~\ref{fig:service_time_of_UE}. This behavior indicates that the MASL algorithms converge to a stable point, suggesting the achievement of the NE for both 
$\mathcal{G}^{\mathrm{S1}}$ and $\mathcal{G}^{\mathrm{S2}}$. 
Since $\mathcal{G}^{\mathrm{S0}}$ is decomposed into these two subgames, the NE of $\mathcal{G}^{\mathrm{S0}}$ is also attained.

\textcolor{black}{Next, we investigate the impact of the learning rates $\alpha$ and $\beta$ in Alg.~\ref{alg:MASL_for_UA} and Alg.~\ref{alg:MASL_for_CO}, respectively. 
	We vary $\alpha$ and $\beta$ over the values $\{0.05, 0.1, 0.2, 0.5, 0.8\}$. 
	Fig.~\ref{fig:learning_rate_impact} shows that, for all tested learning-rate settings, the total service time of all UEs first decreases and then gradually stabilizes, which confirms the convergence behavior of the MASL algorithms. 
	Moreover, Fig.\ref{fig:learning_rate_impact} reveals a clear tradeoff between convergence speed and final performance. 
	Specifically, when the learning rates are relatively large (e.g., $\alpha=\beta=0.8$), 
	accelerate initial convergence, leading to a suboptimal final service time. Conversely, moderate learning rates (e.g., $0.1$ or $0.2$) yield superior final performance. However, excessively small rates (e.g., $0.05$) result in overly conservative updates, which may hinder the algorithm's ability to reach the optimal solution within a finite number of iterations.
	Based on these results, the learning rates can be adjusted according to different practical requirements.}


\subsection{Performance Comparison}
In this section, we present a detailed performance comparison of the MASL algorithm with the benchmark solutions. This comparison is conducted under different system configurations, such as varying the number of APs, bandwidth, ES, and other factors.
\subsubsection{Performance vs Number of APs} 
Fig.~\ref{fig:Num_AP} shows the comparative performance of the baseline solutions as the number of APs increases. As the number of APs grows, the total service time for MASL, BR, mxFP, Selfish, and RARO decreases, due to reduced communication resource competition and improved channel load density. Our MASL algorithm consistently outperforms the other solutions, achieving the lowest total service time in all cases. Specifically, compared with the BR, mxFP, Selfish, and the RARO solution, the MASL solution can reduce the total service time by up to $6.23\%$, $11.06\%$, $67.13\%$, and $44.68\%$ respectively.

\subsubsection{Performance vs APs Bandwidth}
Fig.~\ref{fig:BW_of_AP} illustrates the impact of AP bandwidth on the total service time. The results show that the total service time decreases as the bandwidth of APs increases. A larger bandwidth reduces UE data transmission time, leading to a shorter total service time. 
Compared with the baselines BR, mxFP, Selfish, and RARO, 
MASL can reduce the total service time by up to $5.57\%$, $13.53\%$, $66.96\%$, and $44.65\%$, respectively.

\subsubsection{Performance vs Number of ES}
Fig.~\ref{fig:Num_ES} presents the performance of the baseline solutions as the number of ESs increases. As expected, a higher number of ESs reduces the total service completion time for all solutions by enabling more efficient load distribution. The MASL algorithm achieves the best performance, reducing total service time by up to $28.63\%$, $32.52\%$, $77.92\%$, and $51.44\%$ respectively.

\subsubsection{Performance vs ES Computation Capacity}
Fig.~\ref{fig:comp_cap_of_ES} demonstrates the total service time for different ES computation capacities. Increasing the computation capacity of ES reduces the total service time for all solutions, as higher capacity speeds up computation processing. 
Compared with the BR, mxFP, Selfish, and the RARO solution, the MASL solution can reduce the total service time by up to $16.50\%$, $27.05\%$, $66.87\%$, and $51.85\%$ respectively.

\subsubsection{Performance vs Computation Intensity of Inference Step}
Fig.~\ref{fig:TFLOPs_each_IFstep} shows how total service time varies with the computation intensity of one inference step. As the computational load per inference step increases, the total service time also increases for all solutions. 
MASL consistently achieves the lowest total service time, reducing it by up to $6.29\%$, $23.52\%$, $71.78\%$, and $49.56\%$ respectively.

\subsubsection{Performance vs Number of UEs}
Fig.~\ref{fig:Num_UE} illustrates the total service completion time under varying numbers of UEs. 
The total service completion time increases with the number of UEs, as more UEs lead to heightened resource contention and higher per-task latency under fixed computational capacity. 
Compared with BR, mxFP, Selfish, and RARO, the MASL algorithm reduces the total service completion time by up to $11.54\%$, $19.24\%$, $69.16\%$, and $51.09\%$, respectively.

\begin{figure}[t]
	\centering
	\includegraphics[width=0.83\linewidth]{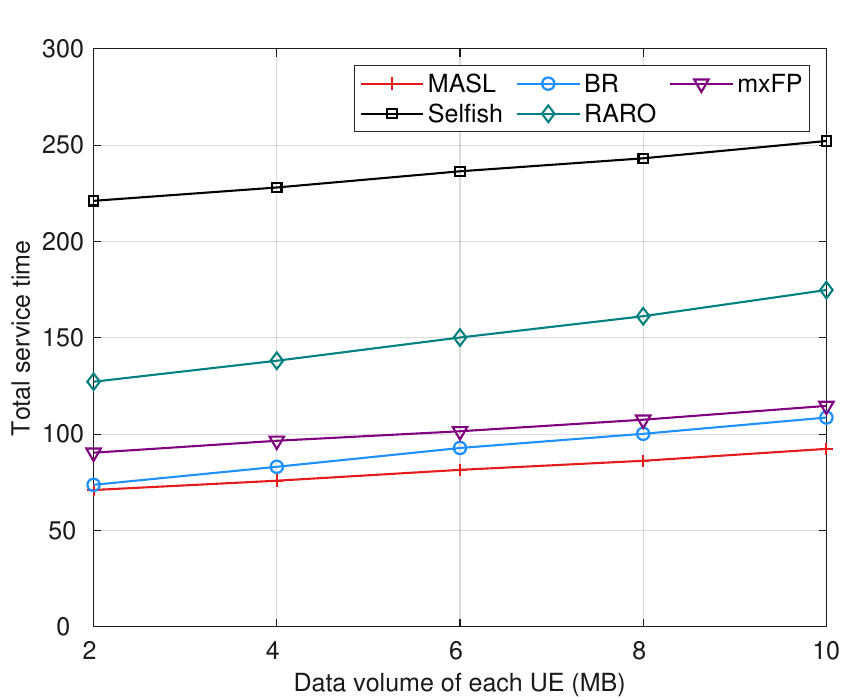}
	\vspace{-3mm}
	\caption{Total service time vs. transmission data volume of UEs.}
	\label{fig:Dn_of_UE}
	\vspace{-3mm}
\end{figure}

\subsubsection{Performance vs UE Transmission Data}
Fig.~\ref{fig:Dn_of_UE} depicts the total service time under varying UE transmission data volumes. 
As the data volume per UE increases, the total service time grows for all five schemes, primarily due to longer transmission durations. 
Compared with BR, mxFP, Selfish, and RARO, the MASL algorithm reduces the total service time by up to $14.96\%$, $19.45\%$, $63.37\%$, and $47.16\%$, respectively.

\subsubsection{The Impacts of Information and Strategies}
Across the scenarios in Fig.~\ref{fig:Num_AP}-\ref{fig:Dn_of_UE}, the proposed MASL framework consistently achieves the lowest total service time. 
This stems from its joint optimization of transmission and computation delays while respecting inference accuracy constraints.

In contrast, the Selfish scheme---where each UE independently selects the nearest AP and the most powerful ES, leads to severe congestion and load imbalance, degrading system-wide performance. 
RARO mitigates this by distributing load more evenly across resources, achieving better statistical balance despite lacking strategic coordination.

BR outperforms mxFP because it directly optimizes strategies using real-time knowledge of others’ actions, whereas mxFP relies on estimated mixed strategies, resulting in slower adaptation. 
However, BR  is often unrealistic in dynamic environments and can induce oscillations due to its deterministic updates and lack of exploration.

By contrast, MASL operates solely on local environmental feedback, using probabilistic strategy updates that enable persistent exploration and ensure robust convergence in stochastic settings. 
Remarkably, even without global information, MASL outperforms BR---demonstrating that the latter’s access to complete information not only fails to guarantee superiority but may exacerbate instability through myopic best-response dynamics. 
Thus, MASL achieves both lower latency and greater convergence stability across   benchmark solutions.

\section{Conclusion}\label{section:Conclusion}
In this work, we investigated the joint communication association and computation offloading problem in MEC-enabled AIGC networks, where GDMs were deployed at ES to satisfy the stringent latency and accuracy requirements of heterogeneous mobile UEs. 
We formulated the UE decision-making process as a Joint Communication Association and Computation Offloading (JCACO) game, in which each UE independently selects its AP, ES, and number of inference steps to minimize service completion time while satisfying accuracy constraints.
We proved that the JCACO game constituted a potential game under both complete and stochastic information scenarios, which guaranteed the existence of an NE.
We developed a distributed Multi-Agent Stochastic Learning (MASL) algorithm and proved its convergence to the NE with strict performance guarantees. 
Extensive simulations demonstrated that MASL significantly outperformed existing benchmark methods, reducing total service completion time by up to $77.92\%$ while satisfying stringent inference accuracy requirements. These results verified the effectiveness, adaptability, and practicality of the proposed approach in real-world MEC-supported AIGC systems.

\begin{IEEEbiography}[{\includegraphics[width=1in,height=1.25in,clip,keepaspectratio]{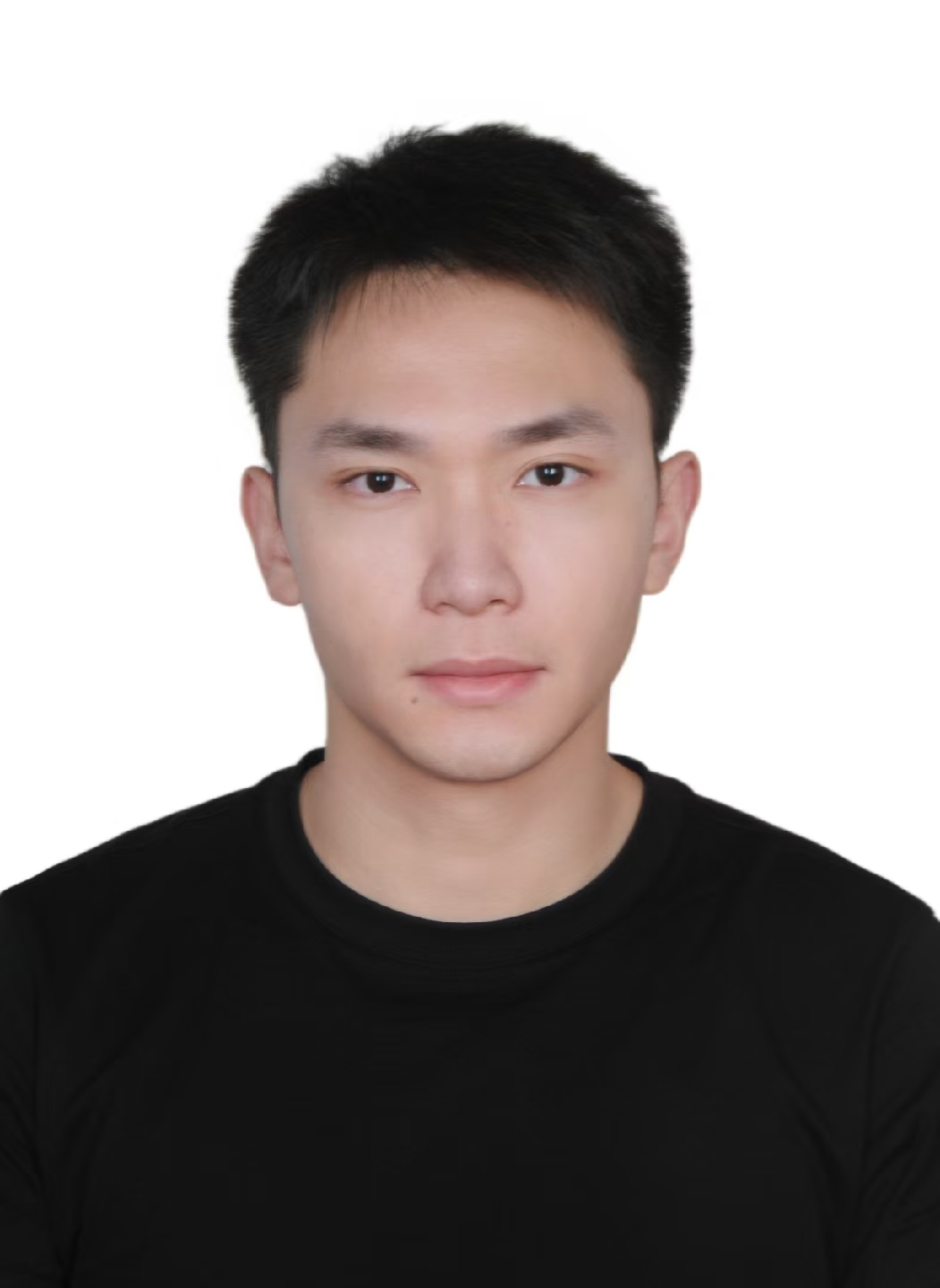}}]{Huaizhe Liu} (Graduate Student Member, IEEE) is currently pursuing the Ph.D. degree with the School of Information Science and Technology, Harbin Institute of Technology, Shenzhen, China. 
He received the M.S. degree in Information and Communication Engineering from Harbin Institute of Technology, Shenzhen, China, in 2020. 
His main research interests are in the interdisciplinary area between stochastic network optimization, algorithmic game theory, and artificial intelligence.
\end{IEEEbiography}

\begin{IEEEbiography}[{\includegraphics[width=1in,height=1.25in,clip,keepaspectratio]{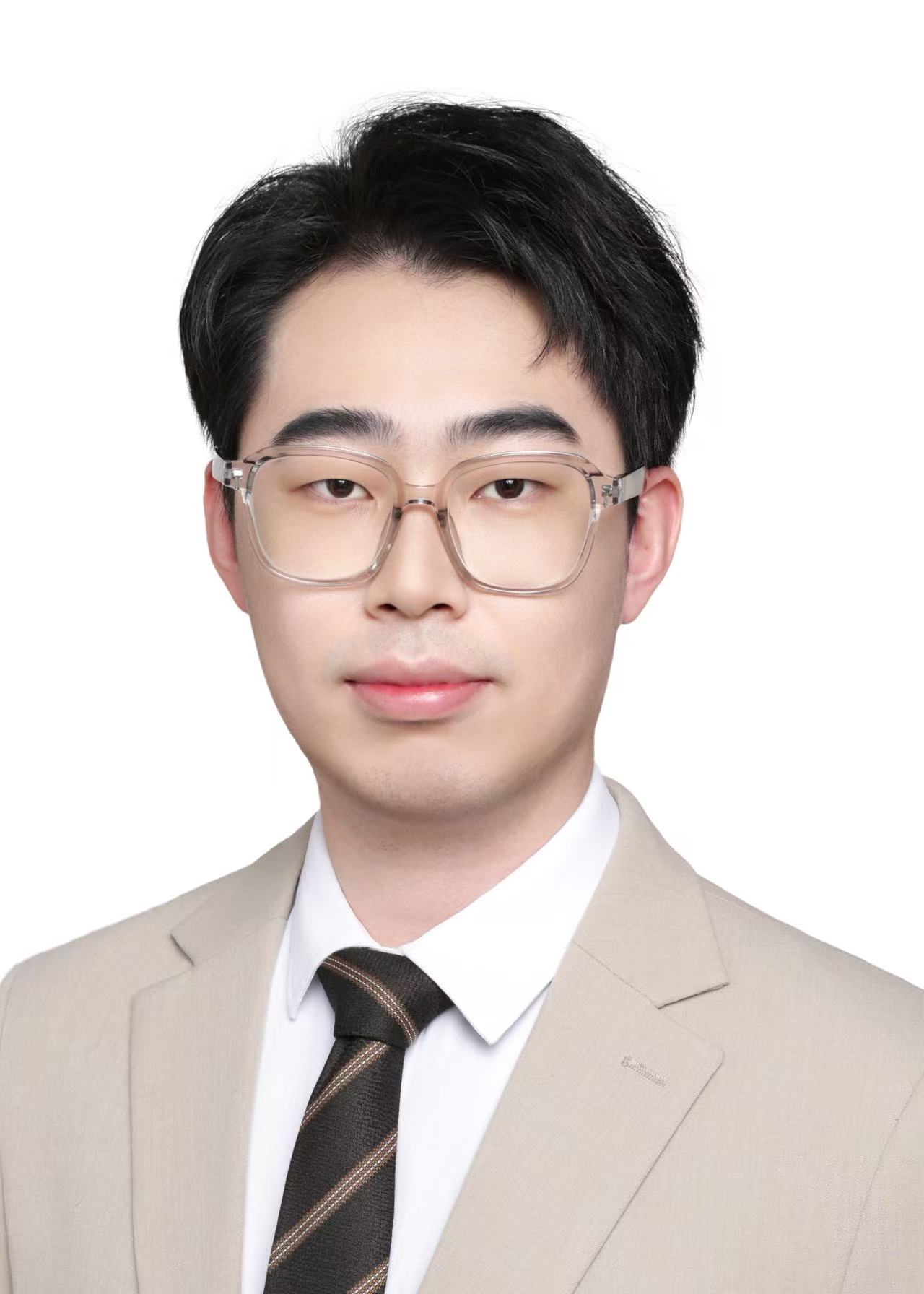}}]{Xinyi Zhuang} (Graduate Student Member, IEEE) is pursuing his Ph.D. degree with the School of Information Science and Technology, Harbin Institute of Technology, Shenzhen, China. He received his B.Eng. degree in Communication Engineering from Northwestern Polytechnical University in 2023. His main research interests include distributed training and inference of large AI models, mobile edge computing, network optimization, and multi-agent reinforcement learning.
\end{IEEEbiography}

\begin{IEEEbiography}[{\includegraphics[width=1in,height=1.25in,clip,keepaspectratio]{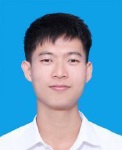}}]{Jiaqi Wu} (Member, IEEE) 
is currently a Lecturer with the School of Computer Science, Guangdong University of Finance, Guangzhou, China. 
He received the Ph.D. degree in Information and Communication Engineering from Harbin Institute of Technology, Shenzhen, China, in 2025.  
His main research interests include network optimization, deep reinforcement learning, edge intelligence, generative AI, and large language model. 
\end{IEEEbiography}

\begin{IEEEbiography}[{\includegraphics[width=1in,clip,keepaspectratio]{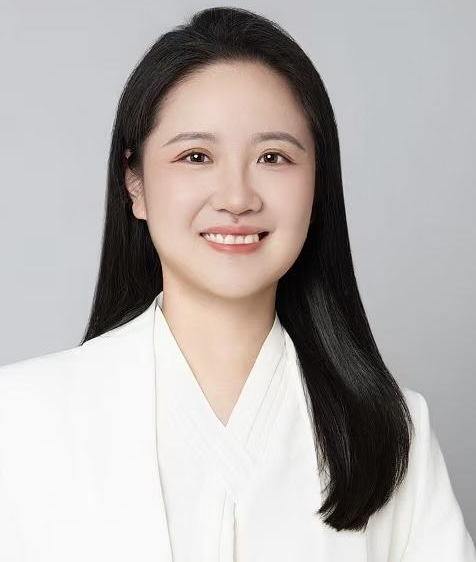}}]{Yuan Luo} (Senior Member, IEEE) received the Ph.D. degree from The Chinese University of Hong Kong, Shenzhen, China, in 2015. She was a Visiting Scholar with the University of California, Berkeley, Berkeley, CA, USA, from 2014 to 2015, a Post-Doctoral Researcher with The Chinese University of Hong Kong, from 2015 to 2017, and a Research Fellow with Imperial College London, London, U.K., from 2017 to 2021. She is currently an Assistant Professor and a Presidential Young Fellow with the School of Science and Engineering, The Chinese University of Hong Kong, Shenzhen, China. Her research interests include artificial intelligence, mechanism design, and low-latency transmission. She was also a recipient of the Best Paper Award at the IEEE WiOpt 2014.
\end{IEEEbiography}

\begin{IEEEbiography}[{\includegraphics[width=1in,clip,keepaspectratio]{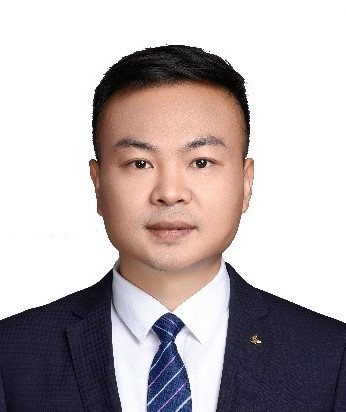}}]{Bin Cao}
(Member, IEEE)  is currently an Associate Professor with the School of Information Science and Technology, Harbin Institute of Technology, Shenzhen, China. 
He received the Ph.D. degree in Information and Communication Engineering from Harbin Institute of Technology, Shenzhen, China, in 2013.
From 2010 to 2012, he was a Visiting Scholar with the University of Waterloo, Canada. 
His research interests include signal processing for wireless communications, cognitive radio networking, and resource allocation for wireless networks.
\end{IEEEbiography}

\begin{IEEEbiography}[{\includegraphics[width=1in,height=1.25in,clip,keepaspectratio]{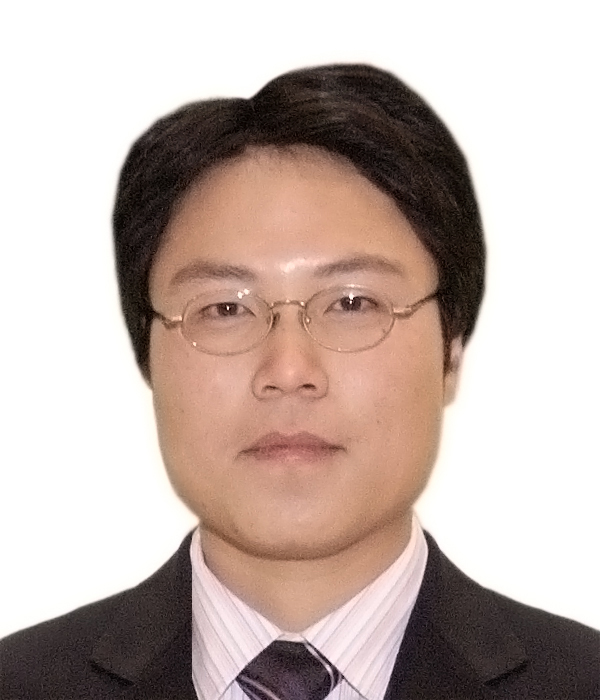}}]{Lin Gao} (Senior Member, IEEE) is a Professor at the School of Information Science and Technology, Harbin Institute of Technology, Shenzhen, China. He received the Ph.D. degree in Electronic Engineering from Shanghai Jiao Tong University, Shanghai, China, in 2010. His main research interests are in the interdisciplinary area between optimization, game theory, and artificial intelligence, with particular focuses on reinforcement learning, federated learning, crowd/edge intelligence, mobile edge computing,  and cognitive networking. He is the co-recipient of 5 Best Paper Awards from leading conference proceedings on wireless communications and networking. He received the IEEE ComSoc Asia-Pacific Outstanding Young Researcher Award in 2016.
\end{IEEEbiography}

\newpage

\appendix

\section{Appendix}
\subsection{Proof for Theorem 1}
\textbf{Theorem 1.} 
	The game $\G^{\mathrm{C1}}$ is a potential game with the following potential function:
	\begin{equation} \small \label{potential_func_AP}
		\begin{aligned}
			\Phi(\mathbf{X};\boldsymbol{\omega}) =\sum_{m=1}^{M} \varphi^{L_{m}(\mathbf{X};\boldsymbol{\omega})},
		\end{aligned}
	\end{equation}
	where $\varphi$ is any value larger than a threshold $\varphi_0 =\sqrt[\epsilon]{2}$, where $\epsilon$ is the minimum time scale of communication system.

\begin{IEEEproof}
	Suppose that an active UE $i\in \N_{A}$ updates its decision $\boldsymbol{x}_{i}$ to $\boldsymbol{x}'_{i}$ (i.e., $x_{ij}=1 \rightarrow x_{ij'}=1$), that is, UE $i$ switches from an AP $j$ to a new one $j'$.
	Let $\mathbf{X}$ and $\mathbf{X'}=(\boldsymbol{x}_{i'},\mathbf{X}_{-i})$ denote the strategy profiles before and after the strategy update of UE $i$, respectively.
	It is easy to see that the load of AP $j$ changes from $L_{j}(\mathbf{X};\boldsymbol{\omega})$ to $L_{j}(\mathbf{X'};\boldsymbol{\omega})=L_{j}(\mathbf{X};\boldsymbol{\omega})-T_{ij}^{\mathrm{Acc}}$ due to the leaving of UE $i$, while the load of AP $j'$ changes from $L_{j'}(\mathbf{X};\boldsymbol{\omega })$ to $L_{j'}(\mathbf{X'};\boldsymbol{\omega})=L_{j'}(\mathbf{X};\boldsymbol{\omega})+T_{ij}^{\mathrm{Acc}}$ due to the  joining of UE $i$.
	
	According to (\ref{transmission_delay}), (\ref{load_AP}), and (\ref{acc_delay}), the change of UE $i$'s utility can be derived as:
	\begin{equation} \small \label{change_utility}
		\begin{aligned}
			\Delta T_{i}^{\mathrm{Acc}}(\mathbf{X} \rightarrow \mathbf{X'};\boldsymbol{\omega}) = L_{j'}(\mathbf{X};\boldsymbol{\omega})+T_{ij'}^{\mathrm{Acc}}-L_{j}(\mathbf{X};\boldsymbol{\omega}).
		\end{aligned}
	\end{equation}
	The change of potential function $\Phi(\mathbf{X};\boldsymbol{\omega})$ can be derived as:
	\begin{equation} \small \label{change_phi_comm}
		\begin{aligned}
			&\Delta\Phi(\mathbf{X \rightarrow X'};\boldsymbol{\omega})  \\ &=\varphi^{L_{j}(\mathbf{X;\boldsymbol{\omega}})-T_{ij}^{\mathrm{Acc}}} + \varphi^{L_{j'}(\mathbf{X};\boldsymbol{\omega})+T_{ij'}^{\mathrm{Acc}}} - \varphi^{L_{j}(\mathbf{X};\boldsymbol{\omega})} - \varphi^{L_{j'}(\mathbf{X};\boldsymbol{\omega})}.
		\end{aligned}
	\end{equation}
	Suppose that $\Delta T_{i}^{\mathrm{Acc}}(\mathbf{X} \rightarrow \mathbf{X'};\boldsymbol{\omega})<0$, that is, 
	\begin{equation} \small 
		\begin{aligned}
			L_{j'}(\mathbf{X};\boldsymbol{\omega})+T_{ij'}^{\mathrm{Acc}} < L_{j}(\mathbf{X};\boldsymbol{\omega}).
		\end{aligned}
	\end{equation}
	Let $\epsilon$ denote the minimum time scale (for resource scheduling) of wireless communication system, e.g., 1 millisecond for 5G NR system, such that: 
	\begin{subequations}
		\begin{numcases}{}
			L_{j'}(\mathbf{X};\boldsymbol{\omega})+T_{ij'}^{\mathrm{Acc}} + \epsilon \leq L_{j}(\mathbf{X};\boldsymbol{\omega}), &\\
			L_{j}(\mathbf{X};\boldsymbol{\omega}) - T_{ij}^{\mathrm{Acc}} + \epsilon \leq L_{j}(\mathbf{X};\boldsymbol{\omega}). &
		\end{numcases}
	\end{subequations}
	
	This implies that: 	
	\begin{subequations} \label{phi_comm_change}
		\begin{numcases}{}
			\varphi^{L_{j'}(\mathbf{X};\boldsymbol{\omega}) + T_{ij'}^{\mathrm{Acc}}} \leq \varphi^{L_{j}(\mathbf{X};\boldsymbol{\omega}) - \epsilon}, &\\
			\varphi^{L_{j}(\mathbf{X};\boldsymbol{\omega}) - T_{ij}^{\mathrm{Acc}}} \leq \varphi^{L_{j}(\mathbf{X};\boldsymbol{\omega}) - \epsilon}. &
		\end{numcases}
	\end{subequations}
	
	According to (\ref{change_phi_comm}) and (\ref{phi_comm_change}), we further have:
	\begin{equation} \small 
		\begin{aligned}
			\Delta\Phi(\mathbf{X} \rightarrow \mathbf{X'};\boldsymbol{\omega})&\leq 2\cdot\varphi^{L_{j}(\mathbf{X};\boldsymbol{\omega}) - \epsilon} - \varphi^{L_{j}(\mathbf{X};\boldsymbol{\omega})} - \varphi^{L_{j'}(\mathbf{X};\boldsymbol{\omega}) }, \\
			&= \varphi^{L_{j}(\mathbf{X};\boldsymbol{\omega})} \cdot (2 \cdot \varphi^{-\epsilon}-1) - \varphi^{L_{j'}(\mathbf{X};\boldsymbol{\omega})}.
		\end{aligned}
	\end{equation}
	
	Obviously, when $2 \cdot \varphi^{-\epsilon}-1 \leq 0$, we have $\Delta\Phi(\mathbf{X} \rightarrow \mathbf{X'};\boldsymbol{\omega})<0$.
	That is, when $\varphi\ge \sqrt[\epsilon]{2}$, we have:
	$\mathrm{sgn}(\Delta T_{i}^{\mathrm{Acc}}(\mathbf{X} \rightarrow \mathbf{X'};\boldsymbol{\omega})) = \mathrm{sgn}(\Delta\Phi(\mathbf{X} \rightarrow \mathbf{X'};\boldsymbol{\omega}))$, which implies that the communication association game is a potential game with potential function $\Phi(\mathbf{X};\boldsymbol{\omega}) =\sum_{m=1}^{M}\varphi^{L_{m}(\mathbf{X};\boldsymbol{\omega})}$, for any $\varphi\ge \sqrt[\epsilon]{2}$. \par
	\end{IEEEproof}
		
\subsection{Proof for Theorem 2}
\textbf{Theorem 2.}
	The game $\G^{\mathrm{C2}}$ is a potential game with the following potential function:
	\begin{equation} \small \label{potential_func_offload}
		\begin{aligned}
			\Psi(\mathbf{Y,D};\boldsymbol{\omega}) =\sum_{k=1}^{K}\phi^{I_{k}(\mathbf{Y,D};\boldsymbol{\omega})}, 			
		\end{aligned}
	\end{equation}
	where $\phi$ is any value larger than a threshold $\phi_0=\sqrt[l]{2}$, where $l$ is the minimum time scale of computation system.

\begin{IEEEproof}
	We proof this theorem through three cases:
	\begin{enumerate}
		\item Case1: Active UE $i \in \N_A$ changes its computation offloading strategy while keeps the inference steps unchanged. 
		\item Case2: Active UE $i \in \N_A$ change its inference steps while keeps the computation offloading strategy unchanged. 
		\item Case3: Active UE $i \in \N_A$ changes both the computation offloading strategy and the inference steps.
	\end{enumerate}	
	
	\emph{Case1}: 
	\\
	Suppose that an active UE $i \in \N_{A}$ updates its decision $\boldsymbol{y}_{i}$ to $\boldsymbol{y}'_{i}$ (i.e., $y_{ik}=1\rightarrow y_{ik'}=1$), that is, UE $i$ switches from an ES $k$ to a new one $k'$.
	Let $\mathbf{Y}$ and $\mathbf{Y}'=(\boldsymbol{y}_{i'},\mathbf{Y}_{-i})$ denote the strategy profiles before and after the strategy update of UE $i$, respectively.
	It is easy to see that the load of ES $k$ changes from $I_{k}(\mathbf{Y,D};\boldsymbol{\omega})$ to $I_{k}(\mathbf{Y,D}';\boldsymbol{\omega})-T_{ik'}^{\mathrm{Comp}}$ due to the leaving of UE $i$, while the load of ES $k'$ changes from $I_{k'}(\mathbf{Y,D};\boldsymbol{\omega})$ to $I_{k'}(\mathbf{Y,D};\boldsymbol{\omega}) + T_{ik'}^{\mathrm{Comp}}$ due to the joining of UE $i$.
	
	According to (\ref{compute_delay}), (\ref{load_ES}), and (\ref{comp_delay}), the change of UE $i$'s utility can de derived as:
	\begin{equation} \small \label{change_H}
		\begin{aligned}
			\Delta T_{i}^{\mathrm{Comp}}(\mathbf{Y\rightarrow Y',D;\boldsymbol{\omega}}) 
			=I_{k'}(\mathbf{Y,D};\boldsymbol{\omega}) + T_{ik'}^{\mathrm{Comp}} - I_{k}(\mathbf{Y,D};\boldsymbol{\omega}).
		\end{aligned}
	\end{equation}	
	The change of potential function $\Psi(\mathbf{Y,D};\boldsymbol{\omega})$ can be drived as:
	\begin{equation} \small \label{change_Psi}
		\begin{aligned}
			&\Delta \Psi(\mathbf{Y\rightarrow Y',D;\boldsymbol{\omega}}) \\
			&=  \phi^{ I_{k'}(\mathbf{Y,D}; \boldsymbol{\omega}) + T_{ik'}^{\mathrm{Comp}} } + \phi^{I_{k}(\mathbf{Y,D};\boldsymbol{\omega}) - T_{ik'}^{\mathrm{Comp}} } \\ &-\phi^{I_{k}(\mathbf{Y,D};\boldsymbol{\omega})} -\phi^{I_{k'}(\mathbf{Y,D};\boldsymbol{\omega})}.
		\end{aligned}
	\end{equation}
	Suppose that $\Delta T_{i}^{\mathrm{Comp}}(\mathbf{Y \rightarrow Y',D;\boldsymbol{\omega}}) <0$. That is,
	\begin{equation} \small 
		\begin{aligned}
			I_{k'}(\mathbf{Y,D};\boldsymbol{\omega}) + T_{ik'}^{\mathrm{Comp}} < I_{k}(\mathbf{Y,D};\boldsymbol{\omega}),
		\end{aligned}
	\end{equation}
	Let $l$ denote the minimum time scale (for resource scheduling) of computation system, e.g., 1 millisecond, such that:	
	\begin{subequations} \small
		\begin{numcases}{}
			I_{k'}(\mathbf{Y,D};\boldsymbol{\omega}) + T_{ik'}^{\mathrm{Comp}} + l\leq I_{k}(\mathbf{Y,D};\boldsymbol{\omega}), &\\
			I_{k}(\mathbf{Y,D};\boldsymbol{\omega})-T_{ik}^{\mathrm{Comp}} + l \leq I_{k}(\mathbf{Y,D};\boldsymbol{\omega}). &
		\end{numcases} 
	\end{subequations}	
	
	This implies that:
	\begin{subequations} \small \label{phi_comp_exp_change}
		\begin{numcases}{}
			\phi^{I_{k'}(\mathbf{Y,D};\boldsymbol{\omega}) + T_{ik'}^{\mathrm{Comp}} } \leq \phi^{I_{k}(\mathbf{Y,D};\boldsymbol{\omega}) - l}, &\\
			\phi^{ I_{k}(\mathbf{Y,D};\boldsymbol{\omega}) - T_{ik}^{\mathrm{Comp}} } \leq \phi^{ I_{k}(\mathbf{Y,D};\boldsymbol{\omega}) - l}. &
		\end{numcases}
	\end{subequations}	
	According to \eqref{change_Psi} and \eqref{phi_comp_exp_change}, we further have:
	\begin{equation} \small 
		\begin{aligned}
			&\Delta \Psi(\mathbf{Y\rightarrow Y',D;\boldsymbol{\omega}}) \\ 
			&\leq 2\cdot \phi^{I_{k}(\mathbf{Y,D};\boldsymbol{\omega}) - l} - 
			\phi^{I_{k}(\mathbf{Y,D};\boldsymbol{\omega})} -\phi^{I_{k'}(\mathbf{Y,D};\boldsymbol{\omega})}				\\	
			&= \phi^{I_k(\mathbf{Y,D};\boldsymbol{\omega})} \cdot (2\cdot\phi^{-l}-1) - \phi^{I_{k'}(\mathbf{Y,D};\boldsymbol{\omega})}.
		\end{aligned}
	\end{equation}
	
	Obviously,  when $2 \cdot\phi^{-l} -1 \leq 0$, we have $ \Psi(\mathbf{Y\rightarrow Y',D;\boldsymbol{\omega}}) <0$. That is, when $\phi\ge \sqrt[l]{2}$, we have: $\mathrm{sgn}(\Delta T_{i}^{\mathrm{Comp}}(\mathbf{Y\rightarrow Y',D;\boldsymbol{\omega}})) = \mathrm{sgn}(\Delta\Psi(\mathbf{Y\rightarrow Y',D;\boldsymbol{\omega}}))$, which implies that the computation offloading game is a potential game with potential function $\Psi(\mathbf{Y,D};\boldsymbol{\omega})=\sum_{k=1}^{K}\phi^{I_{k}(\mathbf{Y,D};\boldsymbol{\omega})}$, for any $\phi\ge \sqrt[l]{2}$. 
	
	\emph{Case2}:
	\\ 
	Suppose that an active UE $i \in \N_{A}$ updates its decision $\boldsymbol{d}_{i}$ to $\boldsymbol{d}'_{i}$.
	Let $\mathbf{D}$ and $\mathbf{D'}=(\boldsymbol{d}_{i'},\mathbf{D}_{-i})$ denote the strategy profiles before and after the strategy update of UE $i$, respectively. According to (\ref{compute_delay}), (\ref{load_ES}), and (\ref{comp_delay}), the change of UE $i$'s utility can de derived as:
	\begin{equation} \small 
		\begin{aligned}
			\Delta T_{i}^{\mathrm{Comp}}(\mathbf{Y,D\rightarrow D'};\boldsymbol{\omega}) = I_{k}(\mathbf{Y,D'};\boldsymbol{\omega}) - I_{k}(\mathbf{Y,D};\boldsymbol{\omega}),
		\end{aligned}
	\end{equation}
	The change of potential function $\Psi(\mathbf{Y,D};\boldsymbol{\omega})$ can be drived as:
	\begin{equation} \small 
		\begin{aligned}
			\Delta \Psi(\mathbf{Y,D\rightarrow D'};\boldsymbol{\omega}) = \phi^{I_{k}(\mathbf{Y,D'};\boldsymbol{\omega})} - \phi^{I_{k}(\mathbf{Y,D};\boldsymbol{\omega})}.
		\end{aligned}
	\end{equation}	
	
	Obviously, we have $\mathrm{sgn}(\Delta T_{i}^{\mathrm{Comp}}(\mathbf{Y,D\rightarrow D';\boldsymbol{\omega}})) = \mathrm{sgn}(\Delta\Psi(\mathbf{Y,D\rightarrow D';\boldsymbol{\omega}}))$ and hence conclude that the computation offloading game is a potential game in Case2.\par
	For Case3, by the similar argument in \emph{Case 1} and \emph{Case 2}, we can also demonstrate that:
	\begin{equation} \small 
		\begin{aligned}
			&\mathrm{sgn}(\Delta T_{i}^{\mathrm{Comp}}(\mathbf{Y\rightarrow Y',D\rightarrow D';\boldsymbol{\omega}})) \\ &=\mathrm{sgn}(\Delta\Psi(\mathbf{Y\rightarrow Y',D\rightarrow D';\boldsymbol{\omega}})) 
		\end{aligned}
	\end{equation}
	Combining the results from the three cases above, we can conclude that the computation offloading game $\G^{\mathrm{C2}}$ is a potential game 
	with potential function $\Psi(\mathbf{Y,D};\boldsymbol{\omega}) =\sum_{k=1}^{K}\phi^{I_{k}(\mathbf{Y,D};\boldsymbol{\omega})}$ for any $\phi \ge \sqrt[l]{2}$.	
\end{IEEEproof}

\subsection{Proof for Theorem 3}

\textbf{Theorem 3.}
	The game $\G^{\mathrm{S1}}$ is a stochastic potential game with the following potential function:
	\begin{align}\label{exp_potential_func_trans}
		\overline{\Phi}(\mathbf{X}) = \sum_{m=1}^{M}\tilde{\varphi}^{\overline{L}_{m}(\mathbf{X})},
	\end{align}
	where $\tilde{\varphi}$ is any value larger than a threshold $\tilde{\varphi}=\sqrt[\epsilon]{2}$, where $\epsilon$ is the minimum time scale of communication system.

\begin{IEEEproof}
	Suppose that an UE $i$ updates its decision $\boldsymbol{x}_{i}$ to $\boldsymbol{x}'_{i}$ (i.e., $x_{ij}=1 \rightarrow x_{ij'}=1$), that is, UE $i$ switches from an AP $j$ to a new one $j'$.
	Let $\mathbf{X}$ and $\mathbf{X'}=(\boldsymbol{x}_{i'},\mathbf{X}_{-i})$ denote the strategy profiles before and after the strategy update of UE $i$, respectively. 
	It is easy to see that the expected load of AP $j$ changes from $\overline{L}_{j}(\mathbf{X})$ to $\overline{L}_{j}(\mathbf{X'})=\overline{L}_{j}(\mathbf{X})-p_{i}\cdot T_{ij}^{\mathrm{Acc}}$ due to the leaving of UE $i$, while the load of AP $j'$ changes from $\overline{L}_{j'}(\mathbf{X})$ to $\overline{L}_{j'}(\mathbf{X'})=\overline{L}_{j'}(\mathbf{X})+p_{i}\cdot T_{ij}^{\mathrm{Acc}}$ due to the  joining of UE $i$.
	
	According to \eqref{transmission_delay}, \eqref{exp_acc}, and \eqref{def_exp_load_AP}, the change of UE $i$'s utility can be derived as:
	\begin{equation} \small \label{exp_change_acc}
		\begin{aligned}
			\Delta \overline{T}_{i}^{\mathrm{Acc}}(\mathbf{X\rightarrow X'}) = \overline{L}_{j'}(\mathbf{X}) + p_{i}\cdot T_{ij'}^{\mathrm{Acc}}- \overline{L}_{j}(\mathbf{X}).
		\end{aligned}
	\end{equation}
	The change of the expected potential function $\overline{\Phi}(\mathbf{X})$ can be derived as:
	\begin{equation} \small \label{exp_change_phi}
		\begin{aligned}
			&\Delta \overline{\Phi}(\mathbf{X}) \\ &=\tilde{\varphi}^{\overline{L}_{j}(\mathbf{X})- p_{i}\cdot T_{ij}^{\mathrm{Acc}}} + \tilde{\varphi}^{\overline{L}_{j'}(\mathbf{X})+ p_{i}\cdot T_{ij'}^{\mathrm{Acc}}} - \tilde{\varphi}^{\overline{L}_{j}(\mathbf{X})} - \tilde{\varphi}^{\overline{L}_{j'}(\mathbf{X})}.
		\end{aligned}
	\end{equation}
	Suppose that 
	$\Delta \overline{T}_{i}^{\mathrm{Acc}}(\mathbf{X\rightarrow X'})<0$, that is,
	\begin{equation} \small 
		\begin{aligned}
			\overline{L}_{j'}(\mathbf{X})+ p_{i}\cdot T_{ij'}^{\mathrm{Acc}} < \overline{L}_{j}(\mathbf{X}).
		\end{aligned}
	\end{equation}
	Let $\epsilon$ denote the minimum time scale (for resource scheduling) of wireless communication system, e.g., 1 millisecond for 5G NR system, such that:
	\begin{subequations} \small \label{small_constant_trans_exp}
		\begin{numcases}{}
			\overline{L}_{j'}(\mathbf{X}) + p_{i}\cdot T_{ij'}^{\mathrm{Acc}}+ \epsilon \leq \overline{L}_{j}(\mathbf{X}), & \label{small_constant_trans_exp_1}\\
			\overline{L}_{j}(\mathbf{X}) - p_{i}\cdot T_{ij}^{\mathrm{Acc}} + \epsilon \leq \overline{L}_{j}(\mathbf{X}). &\label{small_constant_trans_exp_2}
		\end{numcases} 
	\end{subequations}
	This implies that,
	\begin{subequations} \small \label{phi_comm_change_exp}
		\begin{numcases}{}
			\tilde{\varphi}^{\overline{L}_{j'}(\mathbf{X}) + p_{i}\cdot T_{ij'}^{\mathrm{Acc}}} \leq \tilde{\varphi}^{\overline{L}_{j}(\mathbf{X}) - \epsilon}, \\
			\tilde{\varphi}^{\overline{L}_{j}(\mathbf{X}) - p_{i}\cdot T_{ij}^{\mathrm{Acc}}} \leq \tilde{\varphi}^{\overline{L}_{j}(\mathbf{X}) - \epsilon}. 
		\end{numcases}
	\end{subequations}
	According to \eqref{exp_change_acc} and \eqref{small_constant_trans_exp}, we further have:
	\begin{equation} \small \label{change_bar_phi}
		\begin{aligned}
			\Delta\overline{\Phi}(\mathbf{X})&\leq 2\cdot\tilde{\varphi}^{\overline{L}_{j}(\mathbf{X}) - \epsilon} - \tilde{\varphi}^{\overline{L}_{j}(\mathbf{X})} - \tilde{\varphi}^{\overline{L}_{j'}(\mathbf{X}) }, \\
			&= \tilde{\varphi}^{\overline{L}_{j}(\mathbf{X})} \cdot (2 \cdot \tilde{\varphi}^{-\epsilon}-1) - \tilde{\varphi}^{\overline{L}_{j'}(\mathbf{X})}. 
		\end{aligned}
	\end{equation}
	
	Obviously, when $2 \cdot \tilde{\varphi}^{-\epsilon}-1 \leq 0$, we have $\Delta\overline{\Phi}(\mathbf{X} \rightarrow \mathbf{X'})<0$.
	That is, when $\tilde{\varphi}\ge \sqrt[\epsilon]{2}$, we have:
	$\mathrm{sgn}(\Delta \overline{T}_{i}^{\mathrm{Acc}}(\mathbf{X} \rightarrow \mathbf{X'})) = \mathrm{sgn}(\Delta\overline{\Phi}(\mathbf{X} \rightarrow \mathbf{X'}))$, which implies that the communication association game is a potential game with expected potential function $\overline{\Phi}(\mathbf{X}) =\sum_{m=1}^{M}\tilde{\varphi}^{\overline{L}_{m}(\mathbf{X})}$, for any $\tilde{\varphi}\ge \sqrt[\epsilon]{2}$. \par
\end{IEEEproof}

\subsection{Proof for Theorem 4}
\textbf{Theorem 4.}
	The game $\G^{\mathrm{S2}}$ is a stochastic potential game with the expected potential function $\overline{\Psi}(\mathbf{Y,D})$, as follows:
	\begin{equation} \small \label{exp_potential_func_offload}
		\begin{aligned}
			\overline{\Psi}(\mathbf{Y,D}) =\sum_{k=1}^{K}\tilde{\phi}^{\overline{I}_{k}(\mathbf{Y,D})}, 			
		\end{aligned}
	\end{equation}
	where $\tilde{\phi}$ is any value larger than a threshold $\tilde{\phi_0}=\sqrt[l]{2}$, where $l$ is the minimum time scale of computation system.

\begin{IEEEproof}
	We proof this theorem through three cases:
	\begin{enumerate}
		\item Case1: UE $i \in \N$ changes its computation offloading strategy while keeps the inference steps unchanged. 
		\item Case2: UE $i \in \N$ changes its inference steps while keeps the computation offloading strategy unchanged. 
		\item Case3: UE $i \in \N$ changes both the computation offloading strategy and the inference step strategy.
	\end{enumerate}	
	
	\emph{Case1}: 
	\\
	Suppose that an UE $i$ updates its decision $\boldsymbol{y}_{i}$ to $\boldsymbol{y}'_{i}$ (i.e., $y_{ik}=1\rightarrow y_{ik'}=1$), that is, UE $i$ switches from an ES $k$ to a new one $k'$. 
	Let $\mathbf{Y}$ and $\mathbf{Y'}=(\boldsymbol{y}_{i'},\mathbf{Y}_{i})$ denote the strategy profiles before and after the strategy update of UE $i$, respectively.
	It is easy to see that the expected load of ES $k$ changes from $\overline{I}_{k}(\mathbf{Y,D})$ to $\overline{I}_{k}(\mathbf{Y',D})- p_{i}\cdot T_{ik'}^{\mathrm{Comp}}$ due to the leaving of UE $i$, while the expected load of ES $k'$ changes from $\overline{I}_{k'}(\mathbf{Y,D})$ to $\overline{I}_{k'}(\mathbf{Y,D}) + p_{i}\cdot T_{ik'}^{\mathrm{Comp}}$ due to the joining of UE $i$.
	
	According to (\ref{compute_delay}), (\ref{exp_comp}), and (\ref{exp_load_ES}), the change of UE $i$'s utility can be derived as:
	\begin{equation} \small \label{change_exp_comp_delay}
		\begin{aligned}
			\Delta \overline{T}_{i}^{\mathrm{Comp}}(\mathbf{Y\rightarrow Y',D}) 
			=\overline{I}_{k'}(\mathbf{Y,D}) + p_{i}\cdot T_{ik'}^{\mathrm{Comp}} - \overline{I}_{k}(\mathbf{Y,D}).
		\end{aligned}
	\end{equation}
	The change of expected potential function $\overline{\Psi}(\mathbf{Y,D})$ can be derived as:
	\begin{multline} \small \label{change_exp_psi}
			\Delta \overline{\Psi}(\mathbf{Y\rightarrow Y',D}) =  \tilde{\phi}^{ \overline{I}_{k'}(\mathbf{Y,D}) + p_{i}\cdot T_{ik'}^{\mathrm{Comp}} } + \tilde{\phi}^{\overline{I}_{k}(\mathbf{Y,D}) - T_{ik'}^{\mathrm{Comp}} } \\ - \tilde{\phi}^{\overline{I}_{k}(\mathbf{Y,D})} -\tilde{\phi}^{\overline{I}_{k'}(\mathbf{Y,D})}.
	\end{multline}
	Suppose that $\Delta \overline{T}_{i}^{\mathrm{Comp}}(\mathbf{Y\rightarrow Y',D})<0$. That is,
	\begin{equation} \small 
		\begin{aligned}
			\overline{I}_{k'}(\mathbf{Y,D}) + p_{i}\cdot T_{ik'}^{\mathrm{Comp}} < \overline{I}_{k}(\mathbf{Y,D}),
		\end{aligned}
	\end{equation}
	Let $l$ denote the minimum time scale (for resource scheduling) of computation system, e.g., 1 millisecond, such that:
	\begin{subequations}  \label{small_constant_comp}
		\begin{numcases}{}
			\overline{I}_{k'}(\mathbf{Y,D}) + p_{i}\cdot T_{ik'}^{\mathrm{Comp}} + l\leq \overline{I}_{k}(\mathbf{Y,D}), \label{small_constant3} &\\
			\overline{I}_{k}(\mathbf{Y,D}) - p_{i}\cdot T_{ik}^{\mathrm{Comp}} + l \leq \overline{I}_{k}(\mathbf{Y,D}). \label{small_constant4} &
		\end{numcases}
	\end{subequations}
	This implies that:
	\begin{subequations}\label{tilde_phi_comp_exp_change}
		\begin{numcases}{}
			\tilde{\phi}^{ \overline{I}_{k'}(\mathbf{Y,D}) + p_{i}\cdot T_{ik}^{\mathrm{Comp}} } \leq \tilde{\phi}^{\overline{I}_{k}(\mathbf{Y,D}) - l}, &\\
			\tilde{\phi}^{ \overline{I}_{k}(\mathbf{Y,D}) - p_{i}\cdot T_{ik}^{\mathrm{Comp}} } \leq \tilde{\phi}^{ \overline{I}_{k}(\mathbf{Y,D}) - l}. &
		\end{numcases}
	\end{subequations}
	According to \eqref{change_exp_psi} and \eqref{tilde_phi_comp_exp_change}, we further have:
	\begin{multline} \small \label{change_bar_psi}
			\Delta \overline{\Psi}(\mathbf{Y\rightarrow Y',D}) \\
			\leq 2\cdot \tilde{\phi}^{\overline{I}_{k}(\mathbf{Y,D})-l}-\tilde{\phi}^{\overline{I}_{k}(\mathbf{Y,D})-l} - \tilde{\phi}^{\overline{I}_{k'}(\mathbf{Y,D})} \\
			= \tilde{\phi}^{\overline{I}_k(\mathbf{Y,D})} \cdot (2\cdot \tilde{\phi}^{-l}-1) - \tilde{\phi}^{\overline{I}_{k'}(\mathbf{Y,D})}.
	\end{multline}
	
	Obviously, when $2\cdot \tilde{\phi}^{-l}-1 \leq 0$, we have $\Delta \overline{\Psi}(\mathbf{Y\rightarrow Y',D})<0$. 
	That is, when $\tilde{\phi}\geq \sqrt[l]{2}$, we have:
	$\mathrm{sgn}(\Delta \overline{T}_{i}^{\mathrm{Comp}}(\mathbf{Y\rightarrow Y',D})) = \mathrm{sgn}(\Delta \overline{\Psi}(\mathbf{Y\rightarrow Y',D}))$,
	which implies that the stochastic computation offloading game is a potential game with expected potential function $\overline{\Psi}(\mathbf{Y,D}) =\sum_{k=1}^{K}\tilde{\phi}^{\overline{I}_{k}(\mathbf{Y,D})}$, for any $\tilde{\phi_0}=\sqrt[l]{2}$.
	
	\emph{Case2}: 
	\\
	Suppose that an UE $i$ updates its inference step $\boldsymbol{d}_{i}$ to $\boldsymbol{d}'_{i}$.
	Let $\mathbf{D}$ and $\mathbf{D'}=(\boldsymbol{d}'_{i},\mathbf{D}_{-i})$ denote the strategy profiles before and after the strategy update of UE $i$, respectively.
	According to (\ref{compute_delay}), (\ref{exp_comp}), and (\ref{exp_load_ES}), the change of UE $i$'s utility can be derived as: 
	\begin{equation} \small \label{change_exp_comp_delay}
		\begin{aligned}
			\Delta \overline{T}_{i}^{\mathrm{Comp}}(\mathbf{Y,D\rightarrow D'}) 
			=\overline{I}_{k}(\mathbf{Y,D'}) - \overline{I}_{k}(\mathbf{Y,D}).
		\end{aligned}
	\end{equation}
	The change of expected potential function $\overline{\Psi}(\mathbf{Y,D})$ can be derived as:
	\begin{equation} \small 
		\begin{aligned}
			\Delta \overline{\Psi}(\mathbf{Y,D\rightarrow D'}) = \tilde{\phi}^{\overline{I}_{k}(\mathbf{Y,D'})} - \tilde{\phi}^{\overline{I}_{k}(\mathbf{Y,D})}.
		\end{aligned}
	\end{equation}
	Obviously, we have $\mathrm{sgn}(\Delta \overline{T}_{i}^{\mathrm{Comp}}(\mathbf{Y,D\rightarrow D'})) = \mathrm{sgn}(\Delta \overline{\Psi}(\mathbf{Y,D\rightarrow D'}))$ and hence the computation offloading game is a potential game in Case2.
	
	For \emph{Case3}, by the similar argument in \emph{Case1} and \emph{Case2}, we can also demonstrate that:
	\begin{equation} \small 
		\begin{aligned}
			\mathrm{sgn}(\Delta \overline{T}_{i}^{\mathrm{Comp}}(\mathbf{Y,D\rightarrow D'})) = \mathrm{sgn}(\Delta\overline{\Psi}(\mathbf{Y,D\rightarrow D'})). 
		\end{aligned}
	\end{equation}
	Combining the results from the three cases above, we can conclude that the computation offloading game $\G^{\mathrm{C2}}$ is a potential game with expected potential function  $\overline{\Psi}(\mathbf{Y,D}) =\sum_{k=1}^{K}\phi^{\overline{I_{k}}(\mathbf{Y,D})}$, for any $\tilde{\phi_0}\geq \sqrt[l]{2}$.
	
\end{IEEEproof}

\subsection{Proof for Lemma 2}

\textbf{Lemma 2.}
	Let $\mathbf{P} \triangleq (\mathbf{p}_1,\dots,\mathbf{p}_N)$ denote the mixed-strategy profile of the game $\G^{\mathrm{S1}}$, with initial state $\mathbf{P}^0 = [\frac{1}{M}]_{N\times M}$. For a sufficiently small learning rate ($\alpha \rightarrow 0$), the sequence of strategy profiles ${ \mathbf{P}^{\tau} }$ generated by the learning update converges weakly to the solution trajectory of the following ODE:
	\begin{equation} \small \label{ODE}
		\frac{d\mathbf{P}}{d\tau} = \mathscr{H}(\mathbf{P}), \qquad \mathbf{P}(0) = \mathbf{P}^{0},
	\end{equation}
	where $\mathscr{H}(\mathbf{P})$ is the expected update direction, defined as the conditional expectation:
	\begin{equation} \small \label{expect_condition_func}
		\mathscr{H}(\mathbf{P}) = \mathbb{E}\left[\mathscr{F}\bigl(\mathbf{P}^{\tau},\mathbf{X}^{\tau},\tilde{\mathbf{r}}^{\tau}\bigr) \big| \mathbf{P}^{\tau} = \mathbf{P}\right].
	\end{equation}
	Here, $\mathbf{X}^{\tau}=(\boldsymbol{x}_1^{\tau},\dots,\boldsymbol{x}_N^{\tau})$ is the action matrix at iteration $\tau$, $\tilde{\mathbf{r}}^{\tau}=(\tilde{r}_1^{\tau},\dots,\tilde{r}_N^{\tau})$ is the corresponding reward vector, and the function $\mathscr{F}(\cdot)$ encodes the update rule given in \eqref{unified_update_rule}, i.e., $\mathscr{F}(\mathbf{P}^{\tau},\mathbf{X}^{\tau},\tilde{\mathbf{r}}^{\tau}) = \mathbf{P}^{\tau+1}$.

\begin{IEEEproof}
	The proof follows the established framework of weak convergence for stochastic approximation algorithms, as detailed in Theorem 3.1 of \cite{SLA0}. 
	The procedure consists of three main steps.
	
	First, the decentralized update rule \eqref{unified_update_rule} is cast into the canonical stochastic approximation form: $\mathbf{P}^{\tau+1} = \mathbf{P}^{\tau} + \alpha [ \mathscr{H}(\mathbf{P}^{\tau}) + \mathbf{M}^{\tau+1} ]$, where $\mathscr{H}(\mathbf{P}^{\tau})$ is the conditional expected update given by \eqref{expect_condition_func} and $\mathbf{M}^{\tau+1}$ is a martingale difference noise term with zero conditional mean.
	 
	Then, the convergence of the algorithm is then established by verifying the key conditions required by the fundamental weak convergence theorem \cite{weak_converge}. These conditions are satisfied in our setting: the joint strategy space is compact, the rewards are bounded, the update function is continuous and bounded, and given a fixed strategy profile, the actions and rewards are conditionally independent and identically distributed across iterations. Consequently, as the learning rate $\alpha \rightarrow 0$, the interpolated process of the strategy sequence ${\mathbf{P}^{\tau}}$ converges weakly to the solution of the mean ordinary differential equation (ODE) \eqref{ODE}. 
	
	Finally, the specific expression for the drift function $\mathscr{H}(\mathbf{P})$ in \eqref{expect_condition_func} is derived directly by taking the conditional expectation of the update rule \eqref{unified_update_rule}, following the same computation as presented in eq. (16) of \cite{SLA0}. This completes the proof.
	
\end{IEEEproof}

\subsection{Proof for Lemma 3}

\textbf{Lemma 3.}
	With a sufficiently small step-size of learning rate  $\alpha\rightarrow 0$, Alg.1 converges to a stable stationary point of the ODE given in (\ref{ODE-xx}).
		
\begin{IEEEproof}
	Please refer to Lemma 5 in \cite{SLA3}
\end{IEEEproof}

\subsection{Proof for Lemma 4}

\textbf{Lemma 4.}
The following two statements are both true, and they are equivalent:
\begin{enumerate}
	\item All the stable stationary points of the ODE in (\ref{ODE-xx}) are NE points of $\G^{\mathrm{S1}}$.
	\item All the NE points of $\G^{\mathrm{S1}}$ are stable stationary points of the ODE in (\ref{ODE-xx}).
\end{enumerate}

\begin{IEEEproof}
	Please refer to Theorem 3.2 in \cite{SLA0}.
\end{IEEEproof}

\end{document}